%% file: main.tex
\documentclass[sn-mathphys,Numbered]{sn-jnl}
\makeatletter %
\def\@acknow{}%
\long\def\EarlyAcknow#1 \par{%
\def\@acknow{\abstractfont\abstracthead*{Acknowledgments}%
#1\par}}%
\def\@acknow{}%
\long\def\EarlyAcknow#1 \par{%
\def\@acknow{\abstractfont\abstracthead*{Acknowledgments}%
#1\par}}%

\def\printabstract{\ifx\@acknow\empty\else\@acknow\fi\par%
    \ifx\@abstract\empty\else\@abstract\fi\par}
\makeatother

\usepackage{amsmath,amsbsy}
\usepackage{microtype}
\usepackage{color}
\usepackage{xcolor}
\usepackage{algorithm}
\usepackage[noend]{algpseudocode}
\usepackage{multirow}
\usepackage{paralist}
\usepackage{xspace}
\let\originalleft\left
\let\originalright\right
\renewcommand{\left}{\mathopen{}\mathclose\bgroup\originalleft}
\renewcommand{\right}{\aftergroup\egroup\originalright}

\renewcommand{\Pr}{\operatorname*{\textbf{\textup{Pr}}}}

\DeclareMathOperator*{\EE}{\textbf{\textup{E}}}

\newcommand{\set}[1]{\{ #1 \}}

\newcommand{\lt}{\left}
\newcommand{\rt}{\right}
\newcommand{\md}{\middle}

\DeclareMathOperator{\poly}{\ensuremath{\mathrm{poly}}}
\newcommand{\congest}{\ensuremath{\mathsf{CONGEST}}}

\usepackage[most]{tcolorbox}
\tcbuselibrary{theorems}
\tcolorboxenvironment{theorem}{lower separated=false,
                colback=white!90!gray,
colframe=white, fonttitle=\bfseries,
colbacktitle=white!50!gray,
coltitle=black,
enhanced,
boxed title style={colframe=black},
}
\tcolorboxenvironment{corollary}{lower separated=false,
                colback=white!90!gray,
colframe=white, fonttitle=\bfseries,
colbacktitle=white!50!gray,
coltitle=black,
enhanced,
boxed title style={colframe=black},
}

\newcommand{\ann}[1]{%
\text{\footnotesize(#1)}\quad}

\usepackage{mathtools}
\usepackage{subcaption}
\usepackage{tikz}
\usetikzlibrary{
  graphs,
  fit,
  graphs.standard
}

\newcommand{\local}{\ensuremath{\mathsf{LOCAL}}}
\newcommand{\Inp}{\ensuremath{\mathsf{Inp}}}

\newcommand{\rew}{\ensuremath{\mathsf{Rew}}}
\newcommand{\AllDetect}{\ensuremath{\text{$\mathsf{All}$-$\mathsf{Detect}$}}}
\newcommand{\OneDetects}{\ensuremath{\text{$\mathsf{One}$-$\mathsf{Detects}$}}}
\newcommand{\Found}{\ensuremath{\mathsf{Found}}}
\DeclareMathOperator{\A}{\mathcal{A} }
\newcommand{\ktone}{\ensuremath{\mathsf{KT_1}}}
\newcommand{\ktx}{\ensuremath{\mathsf{KT_\rho}}}

\newcommand{\kt}{\ensuremath{\mathsf{KT}}}
\DeclareMathOperator{\id}{\mathsf{id}}
\DeclareMathOperator{\dist}{\mathsf{dist}}

\DeclareMathOperator{\diam}{\mathsf{Diam}}
\DeclareMathOperator{\Diam}{\mathsf{Diam}}

\DeclareMathOperator{\econn}{\stackrel{\text{$e$}}{\text{---}}}

\newcommand{\ktzero}{\ensuremath{\mathsf{KT_0}}}

\def\peter#1{{\underline{\bf{Peter:}}} { [\color{blue} {#1}]}}
\def\mingming#1{{\underline{\bf{Ming Ming:}}} {\color{purple} [{#1}]}}

\newboolean{short}
\newcommand{\onlyShort}[1]{\ifthenelse{\boolean{short}}{#1}{}}
\newcommand{\onlyLong}[1]{\ifthenelse{\boolean{short}}{}{#1}}
\setboolean{short}{false}

\usepackage{graphicx}%
\usepackage{multirow}%
\usepackage{amsmath,amssymb,amsfonts}%
\usepackage{amsthm}%
\usepackage{mathrsfs}%
\usepackage[title]{appendix}%
\usepackage{xcolor}%
\usepackage{textcomp}%
\usepackage{manyfoot}%
\usepackage{booktabs}%
\usepackage{algorithm}%
\usepackage{algorithmicx}%
\usepackage{algpseudocode}%
\usepackage{listings}%
\theoremstyle{thmstyleone}%
\newtheorem{theorem}{Theorem}%
\newtheorem{lemma}[theorem]{Lemma} 
\newtheorem{corollary}[theorem]{Corollary}
\theoremstyle{thmstyletwo}%

\theoremstyle{thmstylethree}%
\newtheorem{definition}{Definition}%
\newtheorem{claim}{Claim}

\begin{document}
\title[Tight Bounds on the Message Complexity of Distributed Tree Verification]{Tight Bounds on the Message Complexity of Distributed Tree Verification} 

 \author[1]{\fnm{Shay} \sur{Kutten}}
\author[2]{\fnm{Peter} \sur{Robinson}}

\author*[2]{\fnm{Ming Ming} \sur{Tan}}\email{mtan@augusta.edu}

\affil[1]{\orgdiv{Technion - Israel Institute of Technology}, \country{Israel}}

\affil[2]{\orgdiv{School of Computer \& Cyber Sciences}, \orgname{Augusta University}, \orgaddress{\city{Augusta}, \state {Georgia}, \country{USA}}}

\def\peter#1{{\underline{\bf{Peter:}}} { [\color{blue} {#1}]}}
\def\mingming#1{{\underline{\bf{Ming Ming:}}} {\color{purple} [{#1}]}}

\EarlyAcknow{
Shay Kutten was supported in part by grant 2070442 from the ISF.
Peter Robinson was supported in part by National Science Foundation (NSF) grant CCF-2402836.
Ming Ming Tan was supported in part by National Science Foundation (NSF) grant CCF-2348346.
}

\abstract{
We consider the message complexity of  verifying whether a given subgraph of the communication network forms a tree with specific properties both in the  $\kt_\rho$ (nodes know their $\rho$-hop neighborhood, including node $\id$s) and the $\ktzero$ (nodes do not have this knowledge)
models. 
We develop a rather general framework that helps in establishing tight lower bounds for various tree verification problems. We also consider two different verification requirements: namely that \emph{every} node detects in the case the input is incorrect, as well as the requirement that \emph{at least one} node detects.
The results are stronger than previous ones in the sense that
 we assume that each node knows the number $n$ of nodes in the graph (in some cases) or an $\alpha$ approximation of $n$ (in other cases). 
For spanning tree verification, 
we show that the message complexity inherently depends on the quality of the given approximation of $n$:  
We show a tight lower bound of $\Omega(n^2)$ for the case $\alpha \ge \sqrt{2}$ and a much better upper bound (i.e., $O(n \log n)$) when nodes are given a tighter approximation.
On the other hand, even for the case when nodes have perfect knowledge of the network size, our framework yields an $\Omega(n^2)$ lower bound on the message complexity of verifying 
a minimum spanning tree (MST). 
This result holds for randomized algorithms with perfect knowledge of the network size, and even when just one node detects illegal inputs, thus improving over the work of Kor, Korman, and Peleg (2013).
Moreover, it also reveals a polynomial separation between ST verification (when nodes know a sufficiently good network-size approximation) and MST verification. 
For verifying a $d$-approximate BFS tree, 
we show that the same lower bound holds even if nodes know $n$ exactly, however, the lower bound is sensitive to $d$, which is the stretch parameter. First, under the $\ktzero$ assumption, we show a tight message complexity lower bound of $\Omega(n^2)$ in the $\local$ model, when $d \le \frac{n}{2+\Omega\lt( 1 \rt)}$. 
For the $\kt_\rho$ assumption, we obtain an upper bound on the message complexity of $O(n\log n)$ in the $\congest$ model, when $d \ge \frac{n-1}{\max \{2,\rho+1\}}$, and use a novel charging argument to show that $\Omega\left(\frac{1}{\rho}\left(\frac{n}{\rho}\right)^{1+\frac{c}{\rho}}\right)$ messages are required even in the $\local$ model for comparison-based algorithms.
For the well-studied special case of $\ktone$,
we obtain a tight lower bound of $\Omega(n^2)$.
}


\maketitle

\section{Introduction} \label{sec:intro}
Verifying the correctness of a given solution to a graph problem is an important problem with numerous applications.
In this setting, there are $n$ nodes that communicate via message passing over the edges of some arbitrary synchronous communication network $G$.
Certain edges in $G$ are labeled and the labeling of an edge $e$ is part of the initial state of the nodes incident to $e$. 
For example, when considering the verification of a minimum spanning tree (MST), in addition to the weight of an edge $e$, the label could also indicate whether $e$ is part of the MST, whereas for verifying a breadth-first search (BFS) tree, $e$'s label may indicate the direction of the edge in the BFS tree $T$ in addition to whether $e \in T$.

Generally, for a graph verification problem $\mathcal{P}$, we assume that the labels correspond to some (possibly weighted and/or directed) graph structure $L$ of $G$.
After observing the labels of their incident edges, the nodes may exchange messages with their neighbors and, eventually, every node needs to output $1$ (``accept'') if $L$ corresponds to a \emph{legal solution} to $\mathcal{P}$ for the communication network $G$.
On the other hand, if $L$ is \emph{illegal}, we study two different requirements:  
\begin{enumerate}
\item $\AllDetect$: Every node outputs $0$ (``reject'').
\item $\OneDetects$: At least one node outputs $0$, whereas the other nodes may output either $0$ or $1$.
\end{enumerate}

We point out that $\AllDetect$ and $\OneDetects$ have both been assumed in previous works.
The $\OneDetects$ requirement, originally proposed in \cite{afek1990memory}, is the basis of the area of distributed verification, e.g., see \cite{korman2005proof,fraigniaud2013towards,naor2020power}, whereas \cite{kor2013tight}, which is closely related to the results of our work, considers $\AllDetect$.
We refer the reader to the 
survey of \cite{feuilloley2016survey}  for a more thorough survey of these results.

Apart from the spanning tree (ST) and minimum spanning tree (MST) verification problems, we also consider verifying approximate versions of breadth-first search (BFS) trees, which we define next:
Consider a connected graph $G$ and some subgraph $H \subseteq G$.
We define $\dist_{H}(u,v)$ to denote the minimum hop distance of nodes $u$ and $v$, when using only edges in $H$.

\begin{definition} \label{def:dBFS}
  A spanning tree $T$ is a \emph{$d$-approximate Breadth-first Search Tree ($d$-approximate BFS)} of $G$ if, for a designated root node $u_0$, the stretch of the shortest-path distance between $u_0$ and any other node in $T$ is at most $d$. Formally,
  $\dist_{T}(u_0, v) \leq d \cdot \dist_{G}(u_0,v)$ for all nodes $v$ in $G$.
\end{definition}
The input labeling for this problem is similar to that of BFS: it induces a directed subgraph and the labels should indicate, for each node, which one of its ports points to its parent and its children in $T$ (if any). 

We study the graph verification problems in both $\congest$ and $\local$ models~\cite{peleg2000distributed}. In $\congest$ model, each message has a bounded size of  $O(\log n)$ bits,  whereas in $\local$ model, the focus is purely on the impact of locality and the message size is unbounded.
Our studies consider the message complexity of distributed algorithms. 
For time complexity, the work of \cite{sarma2012distributed} proves that checking whether a given set of edges induces a spanning tree requires $\Omega(\sqrt{n} + \diam)$ \footnote{ $\diam$ denotes the diameter of the graph} rounds in the $\congest$ model even for randomized algorithms and in fact, they prove that the same bound holds for a long list of fundamental graph verification and construction problems. 
Since \cite{kor2013tight} also gives a deterministic algorithm that has worst-case complexities of $\tilde O(\sqrt{n} + \Diam)$ rounds \footnote{$\tilde O(.)$ hides polylogarithmic factors in $n$.}, the time complexity of spanning tree verification is completely resolved, up to logarithmic factors. %
While the time complexity of distributed verification problems has been studied extensively in previous works, far less is known about the best possible bounds on the message complexity. 
Interestingly, the work of \cite{kor2013tight} also show that any deterministic distributed algorithm that verifies a minimum spanning tree without any knowledge of the network size and guarantees $\AllDetect$ must send $\Theta(n^2)$ messages.
We emphasize that the knowledge of $n$ often strengthens algorithms and reduces their complexity, so our results that are given in spite of assuming nodes posses such knowledge (or, sometimes, the knowledge of an approximation of $n$) are stronger. 
Moreover, our lower bound technique also applies to randomized algorithms. 

\subsection{Our Contributions} \label{sec:contributions}
We show almost-tight bounds on the message complexity of distributed verification for spanning trees and $d$-approximate breadth-first search trees. All our lower bounds for the clean network model (i.e.\ $\ktzero$) hold in the $\local$ model (and hence also $\congest$ model), whereas the upper bounds work in the $\congest$ model. We study the problems in the setting where nodes are equipped with $\alpha$-approximation $\tilde{n}$ of the network size $n$, i.e., $\tilde{n} \in [n/\alpha,\alpha n]$.
For the special case where $\alpha = 1$, the nodes have exact knowledge of the network size. 

\begin{itemize}
\item In Section~\ref{sec:general_kt0_lb}, we present a general lemma (see Lemma \ref{lem:general_kt0_lb}) for deriving message complexity lower bounds for graph verification problems under the $\ktzero$ assumption, where nodes are unaware of their neighbors' IDs initially, which may be of independent interest.
\item For \textbf{spanning tree (ST) verification}, we show that the knowledge of the network size $n$ is crucial for obtaining message-efficient verification algorithms:
  \begin{itemize}
  \item When nodes know the \emph{exact} network size, we give a deterministic ST verification algorithm that guarantees the strong $\AllDetect$ property with a message complexity of only $O(n\log n)$ messages (see Theorem~\ref{thm:mst_upper} on page~\pageref{thm:mst_upper}). 
  \item If nodes have an $\alpha$-approximation of $n$, for any $1 < \alpha < \sqrt{2}$, we still obtain $O(n\log n)$ messages, although we can only achieve $\OneDetects$ (i.e., at least one node detects illegal inputs). We prove that this is unavoidable by showing that $\AllDetect$ requires $\Omega(n^2)$ messages for any $\alpha>1$ (see Theorem~\ref{thm:mstv_all} on page~\pageref{thm:mstv_all}).
  \item On the other hand, we show that, when $\alpha \ge\sqrt{2}$, there is no hope for obtaining a message-efficient algorithm that guarantees $\OneDetects$, as we prove that there are graphs with $\Theta(n^2)$ edges, where the message complexity is $\Omega(n^2)$ (see Theorem~\ref{thm:mstv_all} on page~\pageref{thm:mstv_all}).%
  \end{itemize}
\item For \textbf{MST verification}, we show that $\Omega(n^2)$ poses an insurmountable barrier, as it holds for randomized algorithms, when nodes have perfect knowledge of the network size, and even under $\OneDetects$ (see Corollary~\ref{cor:bfs_kt0_lb} on page~\pageref{cor:bfs_kt0_lb}).
\item For \textbf{verifying a $d$-approximate BFS tree}, we obtain the following results:
    \begin{itemize}
    \item Under the $\kt_0$ assumption, we prove the following:
    \begin{itemize} 
    \item For $d \le \frac{n}{2+\Omega\lt( 1 \rt)}$, any randomized verification algorithm must send $\Omega(n^2)$ messages, even if the nodes have perfect knowledge of the network size (see Theorem~\ref{thm:bfs_kt0_lb} on page~\pageref{thm:bfs_kt0_lb}). 
    \item For $d \ge \frac{n}{2} - \frac{1}{2}$, we give an efficient algorithm that achieves $O(n\log n)$ messages (see Theorem~\ref{thm:bfs_algo} on page~\pageref{thm:bfs_algo}). This shows that the above lower bound is essentially tight in terms of the stretch $d$.
    \end{itemize}
    \item We also consider the $d$-approximate BFS verification problem under the $\kt_\rho$ assumption, for $\rho \ge 1$, where nodes are aware of their $\rho$-hop neighborhood initially, excluding the private random bits of the nodes. Our results hold for the restricted model of comparison-based algorithms that cannot make use of specific values of node IDs. 
    \begin{itemize}
        \item When $d$ is small ($d \le O\lt(\frac{n}{4\rho-2}\rt)$), we show that the lower bound of $\Omega(n^2)$ continues to hold in $\kt_1$. For $\rho \ge 2$, we develop a novel charging argument to show that $\Omega\left(\frac{1}{\rho}\left(\frac{n}{\rho}\right)^{1+\frac{c}{\rho}}\right)$ messages are required, which may turn out to be useful for proving lower bounds for other graph problems in $\kt_\rho$, in particular, for $\rho\ge2$ (see Theorem~\ref{thm:ktx_bfs_lb} on page~\pageref{thm:ktx_bfs_lb}). 
        \item We also show that the restriction on $d$ cannot be improved substantially, by giving an upper bound of $O(n\log n)$ messages when $d \ge \frac{n-1}{\max \{2,\rho+1\}}$ (see Theorem~\ref{thm:bfs_algo} on page~\pageref{thm:bfs_algo}).
    \end{itemize}
    \end{itemize}
\end{itemize} 

\subsection{Additional Related Work} \label{sec:related}
The research on ST, MST and BFS is too vast to provide a comprehensive survey, and thus we limit ourselves to work that is closely related to ours.  
\paragraph{$d$-approximate BFS}
The $d$-approximate BFS tree problem is an important but limited (to a single source) version of the heavily studied spanner concept \cite{peleg1989graph} of a subgraph with a few edges over which the distance between \emph{all} pair of distinct nodes is an approximation of the original distance.
(For $d$-approximate BFS, the stretch requirement applies to each node only with respect to its distance to the root).
In particular, this means that, if $T$ is a $d$-spanner and also a spanning tree, then $T$ must be a $d$-approximate BFS tree. 
However, it is unlikely that $T$ is a $d$-spanner of $G$ due to Erd\H{o}s -Simonovits conjecture that any $d$-spanners with $d<2k+1$ require $\Omega(n^{1+\frac{1}{k}})$ edges in the worst-case. This conjecture was proved for $k = 1, 2, 3, 5$ (see \cite{benson1966minimal, wenger1991extremal}).
The study of BFS approximation (in $\ktzero$) has been motivated by the potential saving in the message and time complexities, especially when compared to those of the Bellman-Ford algorithm; see e.g., \cite{awerbuch1994approximate,nanongkai2014distributed,elkin2004distributed,lenzen2013fast,henzinger2016deterministic}. Also related to $d$-approximate BFS approximation, is the work of \cite{parter2018fault} on fault-tolerant approximate BFS structure, which is a subgraph of the network that contains an approximate BFS spanning tree, after the removal of some set of faulty edges. 

\paragraph{Distributed Verification}
Tarjan~\cite{tarjan1979applications,king1997simpler,harel1985linear,dixon1992verification} considered the question of verifying a minimum weight spanning tree (MST) in the context of centralized computing,
\cite{king1997optimal} addressed the problem in the context of PRAM,
and \cite{korman2006distributed} studied this question in the context of distributed computing with non-determinism, or with pre-processing. 
A verification algorithm may be a part of a fault-tolerant algorithm. Specifically, a verification algorithm can be executed repeatedly. If at some point, the verification fails, then an algorithm for re-computation is activated, followed again by repeated activations of the verification algorithm. In the context of self-stabilization, this was suggested in 
\cite{katz1990self,afek1990memory} (the algorithms here are not self-stabilizing, though). %
More generally, in complexity theory and in cryptography, the issue of the complexity of verifying vs.\ that of computing is one of the main pillars of complexity theory, see, for example, the example of NP-hardness, Zero Knowledge, PCP, and IP (Interactive Proofs). In recent years there has been a lot of research in trying to adapt this kind of theory to distributed computing. We refer the reader to   \cite{korman2005proof,fraigniaud2013towards,naor2020power} for a more thorough survey of these results. 
It seems that the idea to verify a program while it is already running (as opposed to methods such as theorem proving, model checking, or even testing) appeared in general computing possibly after they were studied in distributed computing, but meanwhile, this has become a very developed area, see e.g.\ \cite{leucker2009brief}.  

\paragraph{Message Complexity of Distributed Graph Algorithms}
Understanding the message complexity of distributed algorithms on arbitrary network topologies has received significant attention recently. 
Optimal or near-optimal algorithms are known for several fundamental problems, ranging from designing efficient synchronizers~\cite{DBLP:conf/podc/0001T23} and constructing minimum spanning trees~\cite{king2015construction,DBLP:journals/talg/Pandurangan0S20,DBLP:journals/jacm/Elkin20} to leader election and broadcast~\cite{awerbuch1990trade,DBLP:journals/jacm/KuttenPP0T15,DBLP:conf/wdag/GhaffariK18a}.
The main technique used in the above works for showing lower bounds on the message complexity is by using an edge crossing argument, where the intuition is that nodes cannot distinguish a ``good'' graph from a ``bad'' one unless they communicate over an edge that was crossed. 
However, apart from the previously discussed work of \cite{kor2013tight}, we are not aware of specific results on the message complexity of distributed verification problems.

\subsection{Preliminaries} \label{sec:preliminaries}

\paragraph{Notation} For an integer $n$, we denote the set of $\set{1, 2, \ldots, n}$ by $[n]$. For a graph $H$, we denote the set of vertices of $H$ by $V(H)$. Given a set of vertices $S \subseteq V(H)$, we use the notation $H[S]$ to denote the subgraph induced by the nodes in $S$. 

\paragraph{Computational model}
We consider the standard synchronous $\congest$ and $\local$ models~\cite{peleg2000distributed}, where all nodes are awake initially and communicate via message passing. 
In the $\congest$ model, the bandwidth is limited to a logarithmic number of bits per round over each communication link, whereas there is no such restriction in the $\local$ model.
We assume that each node has a unique ID that is chosen from some polynomial range of integers.

Our main focus of this work is on the \emph{message complexity} of distributed algorithms, which, for deterministic algorithms, is the worst-case number of messages sent over all nodes in any execution.
When analysing randomized algorithms, we assume that each node has access to a private source of unbiased random bits that it may query during its local computation in each round and, in that case, we consider the \emph{expected message complexity}, where the expectation is taken over the private randomness of the nodes.

The initial knowledge of the nodes becomes important when considering message complexity: 
We follow the standard assumptions in the literature, which are $\ktzero$, in the case where nodes do not know the IDs of their neighbors initially.
Under the $\ktzero$ assumption~\cite{awerbuch1990trade}, which is also known as the \emph{clean network model}~\cite{peleg2000distributed}, a node $u$ that has $\delta$ neighbors\footnote{It is thus implicit that the degree of the node is part of its initial knowledge.} also has bidirectional \emph{ports} numbered $1,\dots,\delta$ over which it can send messages.
While $u$ knows its degree, it does not know to which IDs its ports are connected to until it receives a message over this port. 
In contrast, the $\ktone$ assumption ensures that each node knows in advance the IDs of its neighbors and the corresponding port assignments. 
While it takes just a single round to extend $\ktzero$ knowledge to $\ktone$, this would have required $\Omega(m)$ messages in general.  Several algorithms have exploited the additional knowledge provided by $\ktone$ for designing message-efficient algorithms (e.g., \cite{DBLP:conf/icdcn/MashreghiK17,king2015construction,DBLP:conf/wdag/GhaffariK18a,DBLP:conf/wdag/GmyrP18,pai2021can}).

The \emph{state} of a node $u$ consists of its locally stored variables as well as its initial knowledge. After receiving the messages that were addressed to $u$ in the current round, the algorithm performs a \emph{state transition} by taking into account the received messages, $u$'s current state, which may also depend on the node's local random bits when considering randomized algorithms. 
\section{A Framework for Message Complexity Lower Bounds in the $\ktzero$ $\local$ Model} \label{sec:general_kt0_lb}
In this section, we present a general framework for deriving lower bounds on the message complexity of verification problems in the $\ktzero$ $\local$ model.

We remark that the general framework is inspired by the edge crossing technique used in the lower bounds of \cite{DBLP:journals/jacm/KuttenPP0T15,DBLP:conf/wdag/PaiPPR017,pai2021can}, which, however, were designed for specific graph construction and election problems and do not apply to graph \emph{verification}.  
We give some fairly general requirements for a hard 
graph and a corresponding labeling (see Definition~\ref{def:wc_graph}) that, if satisfied, automatically yield nontrivial message complexity lower bounds. 
We start by introducing some technical machinery. 

\noindent\textbf{Rewireable Graphs,  Rewirable Components, and Important Edges.}
Let $H$ be a graph and $L$ be a labelling of $H$. 
We say that an $n$-node graph $H$ is \emph{rewirable} if there exist disjoint subsets $A_1, A_2 \subseteq V(H)$ such that $H[A_1]$ and $H[A_2]$ each contains at least an edge, but there are no edges between $A_1$ and $A_2$. We call $H[A_1]$ and $H[A_2]$ the \emph{rewirable components} of $H$.

In our lower bounds, we identify two \emph{important edges} $e_1=(u_1,v_1) \in H[A_1]$ and $e_2=(u_2,v_2) \in H[A_2]$.
We define $H^{e_1,e_2}$ to be the \emph{rewired graph of $H$ with respect to $e_1$ and $e_2$} on the same vertex set, by removing $e_1$ and $e_2$, and instead connecting $A_1$ and $A_2$ via these four vertices.
Concretely, we have $E(H^{e_1,e_2}) = \lt(E(H) \setminus \set{e_1,e_2}\rt) \cup \set{ (u_1,v_2), (u_2,v_1)}$, whereby $(u_1,v_2)$ and $(u_2,v_1)$ are connected using the same port numbers in $H^{e_1,e_2}$ as for $e_1$ and $e_2$ in $H$.

Next, we present an informal description of our lower bound framework. 
For simplicity, we restrict our discussion here to deterministic algorithms. 
Our approach relies on the construction of a rewirable graph $H$ such that there exists a labeling $L$ which is legal on $H$ but is illegal on any rewired graph of $H$. 
Hence, given the input labeling $L$, if the deterministic verification algorithm is correct, there must exist a node that gives a different output when executed on $H$ than when the execution is on a rewired graph of $H$. 
However, the nodes cannot distinguish the case whether the execution is in $H$ or in $H^{e_1, e_2}$ if the important edges $e_1 \in H[A_1]$ and $e_2 \in H[A_2]$ remain unused throughout the \emph{entire} execution. Consequently, the nodes behave the same in both executions, and therefore produce the same outputs. 
As this would contradict the previous observation that there must exist a node that produces different outputs, it follows that all edges in the rewired components $H[A_1]$ and $H[A_2]$ must be used to send or receive messages, which in turn shows that the message complexity is proportional to the number of edges in the rewirable components. 

Formally, we define $\Inp(G,L,\tilde{n})$ to denote the \emph{input} where we execute the algorithm on graph $G$ with the labeling $L$, and equip all nodes with the network size approximation $\tilde{n}$. 
An input $\Inp(G,L,\tilde{n})$ is said to be \emph{legal} for problem $\mathcal{P}$ if $L$ is a legal solution to $\mathcal{P}$ on the graph $G$.

\medskip
\begin{definition}[Indistinguishability] 
For graphs $H$ and $H'$ where $V(H)=V(H')$, we say that inputs $\Inp(H,L,\tilde{n})$ and $\Inp(H',L',\tilde{n})$ are \emph{indistinguishable for a node $u$}, if, for every deterministic algorithm $\mathcal{A}$, node $u$ has the same state at the start of every round when $\mathcal{A}$ is executed on input $\Inp(H,L,\tilde{n})$, as it does on input $\Inp(H',L',\tilde{n})$. 
\end{definition}

\medskip
Formally, we write $\Inp(H,L,\tilde{n}) \stackrel{}{\cong} \Inp(H',L',\tilde{n})$ if this indistinguishability is true for every node in the graph, and we use the notation $\Inp(H,L,\tilde{n}) \stackrel{S}{\cong} \Inp(H',L',\tilde{n})$ when this holds for every node in some set $S$.

\medskip
\begin{definition}[Hard Base Graph] \label{def:wc_graph}
Consider a verification problem $\mathcal{P}$, and assume that nodes are given an $\alpha$-approximation of the network size.
We say that $H$ is a \emph{hard base graph} for $\mathcal{P}$ if $H$ is rewirable, there is a labeling $L$, a partition of $V(H)$ into vertex sets $S_1$ and $S_2$, such that, for any $\alpha$-approximation $\tilde{n}$ of $|H|$ and any important edges $e_1$ and $e_2$, the following properties hold: %
\begin{compactenum}
\item[(A)] $\Inp(H[S_1],L[S_1],\tilde{n}) \stackrel{S_1}{\cong} \Inp(H,L,\tilde{n}) \stackrel{S_2}{\cong}  \Inp(H[S_2],L[S_2],\tilde{n})$. %
\item[(B)] 
\begin{itemize}
    \item[(i)] For $\AllDetect$, $\tilde{n}$ is an $\alpha$-approximation of $|S_1|$, and $\Inp(H[S_1], L[S_1], \tilde{n})$ is legal for $\mathcal{P}$.
    \item[(ii)] For $\OneDetects$, it must also hold that $\tilde{n}$ is an $\alpha$-approximation of $|S_1|$ and $|S_2|$ (if $S_2\ne\emptyset$).
    Moreover, $\Inp(H[S_1], L[S_1], \tilde{n})$ and $\Inp(H[S_2], L[S_2], \tilde{n})$ must both be legal. 
\end{itemize}

\item[(C)] $\Inp(H^{e_1,e_2}, L, \tilde{n})$ is illegal for $\mathcal{P}$. 
\end{compactenum}
\end{definition}
\medskip

\noindent\textbf{Remarks.} Before stating our lower bound framework based on Definition~\ref{def:wc_graph}, we provide some clarifying comments:
Note that for $H$ to be a hard base graph, it does not need to be an admissible input to problem $\mathcal{P}$. In particular, $H$ can be a disconnected graph even though problem $\mathcal{P}$ is defined for connected networks. 

We emphasize that $S_1$ and $S_2$ are not necessarily related to the rewirable components $A_1$ and $A_2$, introduced above. 
The difference is that $A_1$ and $A_2$ define where the important edges are selected from, whereas $S_1$ and $S_2$ partition $V(H)$ in a way such that $\Inp(H[S_1],L[S_1],\tilde{n})$ or $\Inp(H[S_2],L[S_2],\tilde{n})$ is legal for problem $\mathcal{P}$ (even though $\Inp(H,L,\tilde{n})$ might not actually be an admissible input to problem $\mathcal{P}$).
  
Note that Property~(A) is only relevant when the hard base graph $H$ consists of multiple (i.e., disconnected) components. This will become apparent when proving a lower bound for $d$-approximate BFS tree verification in Section~\ref{sec:bfs_kt0_lb}, where $S_1 = V(H)$ and $S_2 = \emptyset$, which makes the conditions on $H[S_2]$ stated in Property~(A) and Property~(B) vacuously true.

\newcommand{\lemGeneral}{
Consider any $\epsilon$-error randomized algorithm $\mathcal{A}$ for a graph verification problem $\mathcal{P}$, where $\epsilon<\frac{1}{8}$. 
If there exists a hard base graph $H$ for problem $\mathcal{P}$ such that the rewirable components $H[A_1]$ and $H[A_2]$ are both cliques of size $\Theta(n)$,  then $\mathcal{A}$ has an expected message complexity of $\Omega(n^2)$ in the $\ktzero$ $\local$ model. 
}

\medskip
\begin{lemma}[General $\ktzero$ Lower Bound] \label{lem:general_kt0_lb}
\lemGeneral
\end{lemma}

\subsection{Proof of Lemma~\ref{lem:general_kt0_lb}}
\label{app:general_kt0_lb}

We will first show the message complexity lower bound of $\Omega(n^2)$ for deterministic algorithms that may fail on some inputs, when the input graph is sampled from a hard distribution. 
Subsequently, we will use Yao's Minimax Lemma to extend the result to randomized algorithms.  

Assume that there exists a deterministic algorithm $\mathcal{A}$ that has a distributional error of at most $2\epsilon$.
Formally, this means that $2\epsilon$ bounds the probability of sampling an input on which $\mathcal{A}$ fails, from a given input distribution. 
  
\noindent\textbf{The Hard Input Distribution $\mu$:}
We first uniformly at random choose the important edges $e_1, e_2$ from their respective rewirable component.
Then, we flip a fair coin and, if the outcome is heads, we run the algorithm on $\Inp(H^{e_1,e_2},L,\tilde{n})$. 
When the outcome is tails, we distinguish two cases: 
If $S_2= \emptyset$, we execute on $\Inp(H[S_1],L[S_1],\tilde{n})$.
On the other hand, if $S_2\ne\emptyset$, we flip the coin once more and use the outcome to determine whether to execute on    $\Inp(H[S_1],L[S_1],\tilde{n})$ or $\Inp(H[S_2],L[S_2],\tilde{n})$.
By a slight abuse of notation, we abbreviate the last two events as $\Inp(H[S_1])$ and $\Inp(H[S_2])$, respectively. 
Note that conditioning on either  $\Inp(H[S_1])$ or $\Inp(H[S_2])$ fully determines the execution of the deterministic algorithm $\mathcal{A}$. 
In addition, we define $\rew$ (``rewired'') as the event where an input graph is a rewired graph. 
Note that when conditioning on $\rew$, 
the input graph can be any rewired graph $H^{e_1,e_2}$, where $e_1$, $e_2$ are sampled uniformly at random from their respective rewirable component.   %

We define $\Found$ to be the event that a message is sent over an important edge by the algorithm.
According to the distribution $\mu$, we have 
\begin{align}
\Pr\lt[ \rew \rt] &= \tfrac{1}{2}, \label{eq:baserew}
  \\
\Pr\lt[ \Inp(H[S_1]) \rt] &\ge \tfrac{1}{4}, \label{eq:s1prob}\\
\text{if $S_2\ne \emptyset$:} \Pr\lt[ \Inp(H[S_2]) \rt] &= \tfrac{1}{4}. \label{eq:s2prob}
\end{align}

We define $M$ to be the number of messages sent by the algorithm, and
we start our analysis by deriving a lower bound on the expectation of $M$ conditioned on events $\Found$ and $\rew$.
\medskip
\begin{lemma} \label{lem:found}
  $\EE\lt[ M \ \md|\ \Found, \rew \rt] = \Omega\lt( n^{2} \rt)$.
\end{lemma}
\begin{proof}
By assumption, the important edges $e_1=(u_1,v_1)$ and $e_2=(u_2,v_2)$ are chosen uniformly at random such that $u_1, v_1 \in A_1$ and $u_2,v_2 \in A_2$. 
Define $n_i = |A_i|$, for $i\in\set{1,2}$, and recall that $n_i = c_i\,n$, for some constant $c_i>0$, which tells us that the nodes in $A_i$ have $\frac{n_i(n_i-1)}{2} \ge \frac{n_i^2}{4}$ incident edges. 

We can model the task of discovering one of the important edges by the nodes in $A_i$ by the following sampling process:
There is a population $P$ of $\frac{n_i^{2}}{4}$ elements that contains exactly two \emph{bad} elements (i.e., the important edges), whereas the rest are \emph{good} elements.
We sequentially sample (\emph{without} replacement) from $P$ until we obtain a bad element.
The number of trials follows the negative hypergeometric distribution (see, e.g., \cite{feller}), which has an expected value of $\frac{s \cdot f}{|P|-f+1}$, where $|P|$ is the population size, $s$ is the number of desired bad elements after which the process stops, and $f$ refers to the number of good elements in the population.
Plugging in $s=1$, $|P| = \frac{n_i^{2}}{4}$, and $f=\frac{n_i^{2}}{4}-2$, shows that the expected number of trials (i.e., messages sent until a node in $A_i$ finds an important edge) is at least $\frac{n_i^{2}-8}{12} \ge \frac{n_i^{2}}{13}$, for sufficiently large $n$.
\end{proof}

\begin{lemma} \label{lem:foundCor}
If $\EE \lt[ M \rt] = o\lt( n^{2} \rt)$, then the following holds for sufficiently large $n$:
\begin{align}
\Pr\lt[ \Found \ \md|\ \rew \rt] &\le \tfrac{1}{16}. \label{eq:found2}
\end{align}
\end{lemma}

\begin{proof}
We have
\begin{align}
 \EE\lt[ M \rt] 
 &\ge
 \EE\lt[ M \ \md|\ \Found, \rew \rt] 
 \Pr\lt[ \Found, \rew \rt] \notag\\
 &=
 \EE\lt[ M \ \md|\ \Found, \rew \rt] 
 \Pr\lt[ \Found\ \md|\ \rew \rt]\Pr\lt[\rew\rt] \notag\\ 
 \ann{by Lem.~\ref{lem:found}}
 &=
 \Omega\lt( n^2 \rt)
 \cdot
 \begin{multlined}[t]
 \Pr\lt[ \Found \ \md|\ \rew \rt]  
 \Pr\lt[ \rew \rt] \notag 
 \end{multlined}\notag\\ 
 \ann{by \eqref{eq:baserew}}
 &=
 \Omega\lt( n^2 \rt)
 \cdot
 \Pr\lt[ \Found \ \md|\ \rew \rt].  
\notag 
\end{align}
Hence, if $\EE \lt[ M \rt] = o\lt( n^{2} \rt)$, then it must be that $\Pr\lt[ \Found \ \md|\ \rew \rt] = o(1) \le \frac{1}{16}$. 
\end{proof}

Thus, assuming that $\EE\lt[M\rt] = o(n^2)$, Lemma~\ref{lem:foundCor} states that \eqref{eq:found2} holds, and we will show that this leads to a contradiction in the remainder of the proof.

For a given set of vertices $X$, we define the event $Z[X]$ (``Zero output'') as follows:
\begin{enumerate}
    \item For $\OneDetects$, $Z[X]$ occurs if at least one node in $X$ outputs 0. 
    \item For $\AllDetect$, $Z[X]$ occurs if all nodes in $X$ output 0. 
\end{enumerate}

By \eqref{eq:found2}, we have
\begin{align}
  \Pr \lt(Z[V(H)] \mid \rew \rt)  
    &=
        \begin{multlined}[t]
        \Pr \lt(Z[V(H)] \mid \rew, \neg \Found\rt) \Pr(\neg \Found \mid \rew) \\
        + \Pr \lt(Z[V(H)] \mid \rew, \Found \rt)\Pr(\Found \mid \rew)
        \end{multlined}
    \notag \\
    \ann{by \eqref{eq:found2}}
    &\le
        \begin{multlined}[t]
        \Pr \lt(Z[V(H)] \mid \rew, \neg \Found\rt) \Pr(\neg \Found \mid \rew) + \tfrac{1}{16}
        \end{multlined}
   \notag\\ 
    &\le \Pr \lt(Z[V(H)] \mid \rew, \neg \Found\rt)  + \tfrac{1}{16}.  
    \label{eq:g_lb1}
\end{align}
Property~(C) of Definition~\ref{def:wc_graph} tells us that $\Inp(H^{e_1,e_2},L,\tilde{n})$ is illegal for any important edges $e_1, e_2$.

Hence, when conditioning on $\rew$, $\mathcal{A}$ fails unless $Z[V(H)]$ occurs.  
Since $\mathcal{A}$ has a distributional error of at most $2\epsilon$, we have 
\begin{align}
\Pr \lt(\neg Z[V(H)] \mid \rew\rt) \Pr(\rew) 
  \le
\Pr \lt(\mathcal{A} \text{ fails} \mid \rew\rt) \Pr(\rew) 
  \le \Pr \lt(\mathcal{A} \text{ fails} \rt)  
    \leq
        2\epsilon
    \notag
\end{align}
Since $\rew$ happens with probability $\frac{1}{2}$, it follows that
\begin{align}\label{eq:g_lb2_1}
\Pr \lt( Z[V(H)] \mid \rew\rt) > 1 - 4\epsilon.
\end{align}
Combining \eqref{eq:g_lb1} and  \eqref{eq:g_lb2_1} yields 
\begin{align}
\Pr \lt(Z[V(H)]\ \md|\ \rew,\neg \Found\rt) 
  &>  \tfrac{15}{16}-4\epsilon > 0,  \label{eq:g_lb2}
\end{align}
since $\epsilon < \tfrac{1}{8}$.

\medskip
\begin{lemma} \label{lem:rewired}
Consider a graph $H$, a labeling $L$, and a rewired graph $H^{e_1,e_2}$ of $H$.
If $\mathcal{A}$ does not send a message over neither $e_1$ nor $e_2$ when executing on $\Inp(H^{e_1,e_2},L,\tilde{n})$, then it holds that $\Inp(H^{e_1,e_2},L,\tilde{n})  \stackrel{}{\cong} \Inp(H,L,\tilde{n})$. 
In particular, no message is sent over $e_1$ and $e_2$ on input $\Inp(H,L,\tilde{n})$.
\end{lemma}

\begin{proof}
The proof is by induction over the number of rounds, where the basis case follows immediately from the fact that all nodes have identical initial states in both inputs, due to the $\ktzero$ assumption. %

For the inductive step, assume that every node is in the same state at the start of some round $r \ge 1$ when executing on $\Inp(H,L,\tilde{n})$ as well as on $\Inp(H^{e_1,e_2},L,\tilde{n})$.
Since no messages are sent over $e_1$ or $e_2$ and there is no other difference between the two input graphs, it is clear that every node receives the same set of messages over the exact same ports in both executions.
\end{proof}

It follows that the conditioning on $\neg \Found$ (i.e., no important edge is discovered) in \eqref{eq:g_lb2} restricts the sample space of $e_1$ and $e_2$ to only those edges from the rewirable components, over which $\mathcal{A}$ does not send any messages on input  $\Inp(H^{e_1,e_2},L,\tilde{n})$. 
In particular, \eqref{eq:g_lb2} implies that there exists a choice for $e_1$ and $e_2$ such that event $Z[V(H)]$ happens.  
Fix such edges $e_1$, $e_2$. Lemma~\ref{lem:rewired} tells us that the algorithm behaves the same on the inputs $\Inp(H,L,\tilde{n})$ and $\Inp(H^{e_1,e_2},L,\tilde{n})$, i.e.,
\begin{align}
 \text{$Z[V(H)]$ occurs on $\Inp(H,L,\tilde{n})$}. \label{eq:g_lb3}
\end{align}

The following lemma provides a contradiction to \eqref{eq:g_lb3}, which in turn shows that the premise of Lemma~\ref{lem:foundCor} must be false, thus implying the sought bound on the expected message complexity for deterministic algorithms. 

\medskip
\begin{lemma} \label{lem:g_ub}
Event $Z[V(H)]$ does not occur when executing on input $\Inp(H,L,\tilde{n})$.
\end{lemma}
\begin{proof}

We first show the result for $\OneDetects$.
Property~(B) of Definition~\ref{def:wc_graph} ensures that $\Inp(H[S_1],L[S_1],\tilde{n})$ and $\Inp(H[S_2],L[S_2],\tilde{n})$ (if $S_2 \ne \emptyset$) are both legal for $\mathcal{P}$.
From \eqref{eq:s1prob}, it follows that $\Pr\lt( \Inp(H[S_1]) \rt) \ge \tfrac{1}{4}$ and \eqref{eq:s2prob} ensures $\Pr\lt( \Inp(H[S_2]) \rt) \ge \tfrac{1}{4}$ (if $S_2 \ne \emptyset$).

As the distributional error of $\mathcal{A}$ is at most $2\epsilon<\tfrac{1}{4}$, we have the following properties:
\begin{align}
  \text{$Z[V(S_1)]$ does not occur on $\Inp(H[S_1],L[S_1],\tilde{n})$;}
  \label{eq:hs1}\\
  \text{if $S_2 \neq \emptyset$: $Z[V(S_2)]$ does not occur on $\Inp(H[S_2],L[S_2],\tilde{n})$.}
  \label{eq:hs2}
\end{align}
If $S_2=\emptyset$, then \[\Inp(H,L,\tilde{n}) = \Inp(H[S_1],L[S_1],\tilde{n}), \] and we are done due to \eqref{eq:hs1}. Thus, we assume that $S_2\ne\emptyset$.
Recall that $Z[V(H)]$ is true if either $Z[V(S_1)]$ or $Z[V(S_2)]$ are true. Using the indistinguishability guaranteed by Property~(A) of Definition~\ref{def:wc_graph}, we have $\Inp(H[S_1],L[S_1],\tilde{n}) \ \stackrel{S_1}{\cong} \Inp(H,L,\tilde{n})$, and an analogous statement holds for the nodes in $S_2$.
The lemma follows combining \eqref{eq:hs1} and \eqref{eq:hs2}.

Now, for $\AllDetect$, note that $Z[V(H)]$ can only be true if $Z[V(S_1)]$ is true. 
Property~(B) of Definition~\ref{def:wc_graph} tells us that $\Inp(H[S_1],L[S_1],\tilde{n})$ is legal for $\mathcal{P}$. 
Analogously to \OneDetects, it follows from $\Pr\lt( \Inp(H[S_1]) \rt) \ge \tfrac{1}{4}$ that \eqref{eq:hs1} holds.
Using the indistinguishability guaranteed by Property~(A), the lemma follows from \eqref{eq:hs1}.

\end{proof}
We have shown that any deterministic algorithm that has a distributional error of at most $2\epsilon$ when considering the hard input distribution $\mu$ has expected message complexity of $\Omega(n^2)$. 
To complete the proof of Lemma~\ref{lem:general_kt0_lb}, we use Yao's Minimax Lemma, which we restate for completeness. 

\begin{lemma}[Yao's Minimax Lemma, see Prop.~2.6 in \cite{motwani}] \label{lem:yao}
Consider a finite collection of graphs \( \mathcal{I} \) and a distribution \( \Psi \) on \( \mathcal{I} \). 
Let \( X \) be the minimum expected cost of any deterministic algorithm  that succeeds with probability \( 1 - 2\epsilon \) when the graph is sampled according to $\Psi$, for some positive constant \( \epsilon \). Then \( \frac{X}{2} \) lower bounds the expected cost of any randomized algorithm \( R \) on the worst-case graph of \( \mathcal{I} \) that succeeds with probability at least \( 1 - \epsilon \).
\end{lemma}
\medskip

\subsection{A Lower Bound for $d$-Approximate BFS Verification} \label{sec:bfs_kt0_lb}

In this section,  we assume that  the labeling $T$ is an arbitrary directed subgraph of the network $G$. 
The verification task requires checking whether $T$ is a $d$-approximate BFS tree of $G$ (see Def.~\ref{def:dBFS}). Since the subgraph $T$ is directed, each node in $T$ knows its parents and children in $T$. 
We remark that this explicit specification of the direction of the tree (compared to the case where $T$ is an undirected subtree of $G$) can only help the algorithm and hence strengthens our lower bound.  

\begin{theorem} \label{thm:bfs_kt0_lb}
Consider any $\epsilon$-error randomized algorithm that solves $d$-approximate BFS tree verification in the $\ktzero$ $\local$ model, and for any $\epsilon < \frac{1}{8}$.
When $d< \frac{n}{2+\Omega(1)}$, there exists an $n$-node network where the expected message complexity is $\Omega(n^2)$. This holds even when all nodes know the exact network size $n$.
\end{theorem} 

In the remainder of this section, we prove Theorem~\ref{thm:bfs_kt0_lb}.
To instantiate Lemma~\ref{lem:general_kt0_lb}, we define a hard base graph $H$ that satisfies Def.~\ref{def:wc_graph}. 
For the sake of readability and since it does not change the asymptotic bounds, we assume that $2d$ and $\gamma = \frac{n}{2}-d$ are integers. %
We group the vertices of $H$ into $2d+2$ \emph{levels} numbered $0,\dots,2d+1$.
Levels $1$ and $2d+1$ contain $\gamma$ nodes each, denoted by $u_1^{(1)},\dots,u_{\gamma}^{(1)}$ and $u_1^{(2d+1)},\dots,u_{\gamma}^{(2d+1)}$, respectively.
All the other levels consist of only a single node, and we use $u^{(i)}$ to denote the (single) node on level $i \in ([0,2d] \setminus \set{1})$.

Next, we define the edges of $H$. 
Every node on level $i < 2d+1$ is connected via \emph{inter-level} edges to every other node on level $i+1$.
Moreover, the nodes on each level form a clique. 
Figure~\ref{fig:bfs_lb_unwired} gives an example of this construction. 

To ensure that $H$ is rewirable, we set $A_1$ and $A_2$ to be the sets of clique nodes on level $1$ and level $2d+1$, respectively. 
Recall that the number of nodes in $A_1$ (and $A_2$) is $\gamma = \frac{n}{2}-d$. Since $d<\frac{n}{2+\Omega(1)}$, it is straightforward to verify that $\gamma$ is $\Theta(n)$. 
We summarize the properties of $H$ in the next lemma:

\definecolor{myblue}{RGB}{80,80,160}
\definecolor{vcol}{RGB}{196,218,238}
\definecolor{orange}{RGB}{225,113,57}
\definecolor{mygreen}{RGB}{211,238,205}

\begin{figure}[t]
  \centering
\begin{subfigure}[t]{0.45\textwidth}
  \centering
\begin{tikzpicture}[auto,yscale=0.8,scale=0.7]
  \tikzset{vertex/.style={draw,circle,fill=vcol}} %
  \tikzset{vertex-set/.style={fill=mygreen,rounded corners}} %
  \tikzset{label/.style={text=black}} %
  \tikzset{long/.style={decorate,decoration=snake}} %
  \tikzset{edge/.style={draw,black,-,auto}} %
  \tikzstyle{orange_edge}=[very thick,draw=orange,-,auto]
  \pgfdeclarelayer{background}
  \pgfdeclarelayer{inbetween}
  \pgfdeclarelayer{nodelayer}
  \pgfdeclarelayer{edgelayer}
  \pgfsetlayers{background,inbetween,edgelayer,nodelayer,main}

	\begin{pgfonlayer}{nodelayer}
		\node [label,xshift=-0.1cm]  at (3, 13) {$u_1^{(1)}$};
		\node [label,xshift=-0.1cm] at (5, 13) {$u_2^{(1)}$};
		\node [label,xshift=-0.1cm]  at (7, 13) {$u_3^{(1)}$};
		\node [label,xshift=-0.1cm] at (9, 13) {$u_4^{(1)}$};
		\node [label,xshift=-0.1cm] at (3, 11) {$u^{(2)}$};
		\node [label] (vdots) at (3.5, 9) {$\vdots$};
		\node [label,xshift=-0.2cm] at (3, 7) {$u^{(2d)}$};
		\node [label,xshift=-0.1cm] at (3, 15) {$u^{(0)}$};
		\node [style=vertex] (u11) at (3.5, 13) {};
		\node [style=vertex] (u0) at (3.5, 15) {};
		\node [style=vertex] (u12) at (5.5, 13) {};
		\node [style=vertex] (u13) at (7.5, 13) {};
		\node [style=vertex] (u14) at (9.5, 13) {};
		\node [style=vertex] (u2) at (3.5, 11) {};
		\node [style=vertex] (u2d) at (3.5, 7) {};
		\node [style=vertex] (u2d+1-1) at (3.5, 5) {};
		\node [style=vertex] (u2d+1-2) at (5.5, 5) {};
		\node [style=vertex] (u2d+1-3) at (7.5, 5) {};
		\node [style=vertex] (u2d+1-4) at (9.5, 5) {};
		\node [label,yshift=0.2cm,below of=u2d+1-1]  {$u_1^{(2d+1)}$};
		\node [label,yshift=0.2cm,below of=u2d+1-2]  {$u_2^{(2d+1)}$};
		\node [label,yshift=0.2cm,below of=u2d+1-3]  {$u_3^{(2d+1)}$};
		\node [label,yshift=0.2cm,below of=u2d+1-4]  {$u_4^{(2d+1)}$};
	\end{pgfonlayer}
  \begin{pgfonlayer}{edgelayer}
    \draw[orange_edge,->] (u0) to (u11);
    \draw[orange_edge,bend left,->] (u0) to (u12);
    \draw[orange_edge,bend left,->] (u0) to (u13);
    \draw[orange_edge,bend left,->] (u0) to (u14);
    \draw[edge,out=315,in=225] (u11) to (u12);
    \draw[edge,out=45] (u11) to (u13);
    \draw[edge,out=45] (u11) to (u14);
    \draw[edge,out=315,in=225] (u12) to (u13);
    \draw[edge,out=45] (u12) to (u14);
    \draw[edge,out=315,in=225] (u13) to (u14);
    \draw[orange_edge,->] (u11) to (u2);
    \draw[edge,bend left] (u12) to (u2);
    \draw[edge,bend left] (u13) to (u2);
    \draw[edge,bend left] (u14) to (u2);
    \draw[orange_edge,->] (u2) to (vdots);
    \draw[orange_edge,->] (vdots) to (u2d);
    \draw[orange_edge,->] (u2d) to (u2d+1-1);
    \draw[orange_edge,bend right,<-] (u2d+1-2) to (u2d);
    \draw[orange_edge,bend right,<-] (u2d+1-3) to (u2d);
    \draw[orange_edge,bend right,<-] (u2d+1-4) to (u2d);
    \draw[edge,out=315,in=225] (u2d+1-1) to (u2d+1-2);
    \draw[edge,out=45] (u2d+1-1) to (u2d+1-3);
    \draw[edge,out=45] (u2d+1-1) to (u2d+1-4);
    \draw[edge,out=315,in=225] (u2d+1-2) to (u2d+1-3);
    \draw[edge,out=45] (u2d+1-2) to (u2d+1-4);
    \draw[edge,out=315,in=225] (u2d+1-3) to (u2d+1-4);
  \end{pgfonlayer}
\end{tikzpicture}
\caption{\small The hard base graph $H$ for $d$-approximate BFS tree verification}
\label{fig:bfs_lb_unwired}
\end{subfigure}
\phantom{------}
\begin{subfigure}[t]{0.45\textwidth}
  \centering
\begin{tikzpicture}[auto,yscale=0.8,scale=0.7]
  \tikzset{vertex/.style={draw,circle,fill=vcol}} %
  \tikzset{vertex-set/.style={fill=mygreen,rounded corners}} %
  \tikzset{label/.style={text=black}} %
  \tikzset{long/.style={decorate,decoration=snake}} %
  \tikzset{edge/.style={draw,black,auto}} %
  \tikzstyle{orange_edge}=[very thick,draw=orange,-,auto]
  \pgfdeclarelayer{background}
  \pgfdeclarelayer{inbetween}
  \pgfdeclarelayer{nodelayer}
  \pgfdeclarelayer{edgelayer}
  \pgfsetlayers{background,inbetween,edgelayer,nodelayer,main}

	\begin{pgfonlayer}{nodelayer}
		\node [label,xshift=-0.1cm]  at (3, 13) {$u_1^{(1)}$};
		\node [label,xshift=-0.1cm] at (5, 13) {$u_2^{(1)}$};
		\node [label,xshift=-0.1cm]  at (7, 13) {$u_3^{(1)}$};
		\node [label,xshift=-0.1cm] at (9, 13) {$u_4^{(1)}$};
		\node [label] at (3, 11) {$u^{(2)}$};
		\node [label] (vdots) at (3.5, 9) {$\vdots$};
		\node [label,xshift=-0.2cm] at (3, 7) {$u^{(2d)}$};
		\node [label,xshift=-0.1cm] at (3, 15) {$u^{(0)}$};
		\node [style=vertex] (u11) at (3.5, 13) {};
		\node [style=vertex] (u0) at (3.5, 15) {};
		\node [style=vertex] (u12) at (5.5, 13) {};
		\node [style=vertex] (u13) at (7.5, 13) {};
		\node [style=vertex] (u14) at (9.5, 13) {};
		\node [style=vertex] (u2) at (3.5, 11) {};
		\node [style=vertex] (u2d) at (3.5, 7) {};
		\node [style=vertex] (u2d+1-1) at (3.5, 5) {};
		\node [style=vertex] (u2d+1-2) at (5.5, 5) {};
		\node [style=vertex] (u2d+1-3) at (7.5, 5) {};
		\node [style=vertex] (u2d+1-4) at (9.5, 5) {};
		\node [label,yshift=0.2cm,below of=u2d+1-1]  {$u_1^{(2d+1)}$};
		\node [label,yshift=0.2cm,below of=u2d+1-2]  {$u_2^{(2d+1)}$};
		\node [label,yshift=0.2cm,below of=u2d+1-3]  {$u_3^{(2d+1)}$};
		\node [label,yshift=0.2cm,below of=u2d+1-4]  {$u_4^{(2d+1)}$};
	\end{pgfonlayer}
  \begin{pgfonlayer}{edgelayer}
    \draw[orange_edge,->] (u0) to (u11);
    \draw[orange_edge,bend left,->] (u0) to (u12);
    \draw[orange_edge,bend left,->] (u0) to (u13);
    \draw[orange_edge,bend left,->] (u0) to (u14);
    \draw[edge,out=315,in=225] (u11) to (u12);
    \draw[edge,out=45] (u11) to (u13);
    \draw[edge,out=45] (u11) to (u14);
    \draw[edge,out=315,in=225] (u12) to (u13);
    \draw[edge,out=45] (u12) to (u14);
    \draw[edge,myblue,out=315,in=225,very thick,dashed] (u13) to node[below,xshift=-0.3cm] {$e_1$} (u14);
    \draw[orange_edge,->] (u11) to (u2);
    \draw[edge,bend left] (u12) to (u2);
    \draw[edge,bend left] (u13) to (u2);
    \draw[edge,bend left] (u14) to (u2);
    \draw[orange_edge,->] (u2) to (vdots);
    \draw[orange_edge,->] (vdots) to (u2d);
    \draw[orange_edge,->] (u2d) to (u2d+1-1);
    \draw[orange_edge,bend right,<-] (u2d+1-2) to (u2d);
    \draw[orange_edge,bend right,<-] (u2d+1-3) to (u2d);
    \draw[orange_edge,bend right,<-] (u2d+1-4) to (u2d);
    \draw[edge,out=315,in=225] (u2d+1-1) to (u2d+1-2);
    \draw[edge,out=45] (u2d+1-1) to (u2d+1-3);
    \draw[edge,out=45] (u2d+1-1) to (u2d+1-4);
    \draw[edge,out=315,in=225] (u2d+1-2) to (u2d+1-3);
    \draw[edge,draw=myblue,very thick,dashed,out=45] (u2d+1-2) to node[below] {$e_2$} (u2d+1-4);
    \draw[edge,out=315,in=225] (u2d+1-3) to (u2d+1-4);
    \draw[edge,draw=myblue,very thick,out=280,in=80] (u13) to (u2d+1-4);
    \draw[edge,draw=myblue,very thick,out=280,in=80] (u14) to (u2d+1-2);
  \end{pgfonlayer}
\end{tikzpicture}
\caption{\small The rewired graph $H^{e_1,e_2}$ with the (dashed) important edges $e_1$ and $e_2$ that are replaced by the thick blue edges.}
\label{fig:bfs_lb_rewired}
\end{subfigure}

\caption{\small The graphs used in the lower bound for $d$-approximate BFS verification. The  directed orange edges show the labeling, which is a BFS tree, and hence a legal input for any $d$.}
\label{fig:bfs_lb}
\end{figure}
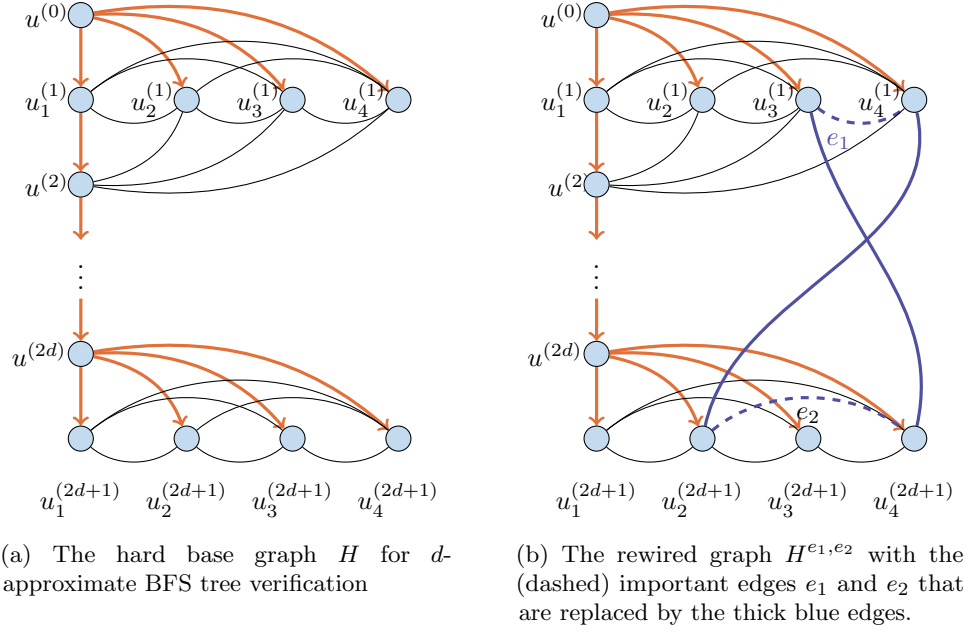

\begin{lemma} \label{lem:bfs_graph}
Graph $H$ has $2\gamma+2d=n$ nodes and $\gamma^2+2\gamma+2d-2=\Theta(\gamma^2)$ edges.
It is a rewirable graph that contains two cliques, each on $\gamma = \Theta(n)$ nodes, as rewirable components.
The directed subgraph $T \subseteq H$ which contains all nodes of $H$, all the inter-level edges from level $0$ to $1$, the edge $(u_1^{(1)}, u^{(2)})$, and all the inter-level edges from level $i$ to $i+1$ for $i \geq 2$, is a $d$-approximate BFS tree of $H$.
(The directed orange edges in Figure~\ref{fig:bfs_lb_unwired} provide an example of $T$). 
\end{lemma}

The next lemma tells us that the properties of Definition~\ref{def:wc_graph}  hold, which will allow us to instantiate Lemma~\ref{lem:general_kt0_lb} for obtaining the sought lower bound.

\begin{lemma} \label{lem:wc_bfs}
Consider the directed tree $T$ defined in Lemma~\ref{lem:bfs_graph}.
Then, $T$ is a $d$-approximate BFS tree of $H$, but not of any rewired graph $H^{e_1,e_2}$ and, consequently, graph $H$ is a hard base graph for $d$-approximate BFS verification with labeling $T$.
This holds even if all nodes know the exact network size.
\end{lemma}
\begin{proof}
We define $S_1 = V(H)$ and thus $S_2 = \emptyset$.
For a given pair of important edges $e_1$ and $e_2$, we need to show Properties~(A), (B), and (C) of Definition~\ref{def:wc_graph}.
Property~(A) is trivial.
For (B), note that $\Inp(H[S_1],L[S_1],n) = \Inp(H,T,n)$, and by Lemma~\ref{lem:bfs_graph}, $T$ is indeed a $d$-approximate BFS tree of $H$. 
For property (C), observe that the input tree $T$ is not a $d$-approximate BFS tree of the rewired graph $H^{e_1,e_2}$, for any $e_1$ and $e_2$, because $\dist_{T}(u^{(0)}, u_{k}^{(2d+1)}) = 2d+1$ for $k\geq 1$, whereas $\dist_{H^{e_1,e_2}}(u^{(0)}, u_{l}^{(2d+1)}) = 2$ for $l$ such that $u_{l}^{(2d+1)}$ is incident to $e_2$.
\end{proof}

Theorem~\ref{thm:bfs_kt0_lb} follows by combining Lemma~\ref{lem:bfs_graph} and Lemma~\ref{lem:wc_bfs} with the general lower bound Lemma~\ref{lem:general_kt0_lb}.

\subsection{A Lower Bound for MST Verification} \label{sec:mst_kt0_lb} 

We now show that the lower bound graph used for $d$-approximate BFS tree verification in Section~\ref{sec:bfs_kt0_lb} is versatile enough to also yield a lower bound for MST verification. 

\begin{corollary} \label{cor:bfs_kt0_lb}
Suppose that all nodes know the exact network size. 
Let $\epsilon < \frac{1}{8}$. 
Any $\epsilon$-error randomized algorithm that verifies whether an input is a minimum spanning tree in the $\ktzero$ $\local$ model sends $\Omega(n^2)$ messages in expectation, even if the algorithm only ensures that at least one node detects illegal inputs (i.e.\ $\OneDetects$).
The same result holds for verifying an approximate MST.

\end{corollary}
\begin{proof} 
We use a variant of the hard base graph $H$ that we employed in the proof of Theorem~\ref{thm:bfs_kt0_lb}, which is shown in Figure~\ref{fig:bfs_lb}.
The main difference is that now we consider an input labeling $L$ that induces an undirected weighted subgraph, and we consider only $5$ layers, where layer $1$ and layer $4$ are cliques, and the other layers consist of a single node each.  Hence, there is a single edge $e=(u^{(2)},u^{(3)})$ that acts as a bridge between the two cliques.
The edge $e$ has weight $W>1$ whereas every other edge in the graph has weight $0$.
The labeling $L$ that we consider induces an MST that contains $e$ and some arbitrary spanning tree on the rest of the graph.
We now show that this satisfies Properties~(A), (B), and (C) of Definition~\ref{def:wc_graph}: 
We choose $S_1=V(H)$ and $S_2=\emptyset$, which makes Property~(A) vacuously true.
For Property~(B), observe that $\Inp(H[S_1],L[S_1],n)$ corresponds to a valid MST of $H$ and hence is legal. 
To see why Property~(C) holds, it is sufficient to observe that any MST of the rewired graph $H^{{e_1,e_2}}$ must include one of the rewired edges (connecting the two cliques) instead of the bridge edge $e$ that has weight $W$. 
\end{proof}

\subsection{A Lower Bound for Spanning Tree Verification}
\label{sec:general_lb}
We now consider the \emph{Spanning Tree (ST) verification} problem where nodes have an $\alpha$-approximation of the network size.
The input is a connected graph $G$, a subgraph $T$ of $G$, and an integer $\tilde{n}$ which is an $\alpha$-approximation to the actual network size. The verification task requires checking distributively whether $T$ is an ST of $G$, i.e., $T$ is a tree that contains all nodes. 

We remark that a message complexity lower bound of $\Omega(n^2)$ assuming $\AllDetect$ was previously shown by Kor, Korman, and Peleg~\cite{kor2013tight} for deterministic algorithms in the setting where the network size is unknown. By using our framework, we generalize the message complexity lower bound to randomized algorithms and to the setting where nodes have an $\alpha$-approximation of the network size. 
We also show that the bound depends on whether we assume $\OneDetects$ or $\AllDetect$.

\begin{theorem} \label{thm:mstv_all}
Consider an $\epsilon$-error randomized algorithm $\mathcal{A}$ that solves spanning tree verification in the $\ktzero$ $\local$ model, where $\epsilon<\frac{1}{8}$, and suppose that nodes know an $\alpha$-approximation of the network size: 
\begin{compactitem}
\item If all nodes detect an illegal input (i.e.\ $\AllDetect$), then the message complexity is $\Omega(n^2)$ in expectation, for any constant $\alpha>1$. 
\item If the algorithm satisfies only $\OneDetects$, then the expected message complexity is still $\Omega(n^2)$, for any constant $\alpha \ge \sqrt{2}$. 
\end{compactitem}
\end{theorem}
We emphasize that the bounds on the approximation ratio in Theorem~\ref{thm:mstv_all} are tight, since in Section~\ref{sec:st_algo} we show that, if $\alpha < \sqrt{2}$, then the message complexity reduces drastically to just $O(n\log n)$ for $\OneDetects$,
and this holds even for $\AllDetect$ in the case where $\alpha = 1$.

Combining the following lemma with Lemma~\ref{lem:general_kt0_lb} immediately yields Theorem~\ref{thm:mstv_all}:

\begin{lemma} \label{lem:wc_graph_mst}
There exists a graph $H$ that is a hard base graph for ST verification with a suitable labeling where the rewireable components are cliques of size $\Omega(n)$, assuming the following restriction on the network size approximation known by the nodes:
\begin{compactenum}
\item[(i)] For $\AllDetect$, this holds for any constant $\alpha > 1$.
\item[(ii)] For $\OneDetects$, this holds for any constant $\alpha \ge \sqrt{2}$.
\end{compactenum}
\end{lemma}

\begin{proof} 
We define a suitable lower bound graph $H$ of size $n$ that is a disconnected graph whose components are two cliques $C$ and $C'$. 
We will specify the sizes of $C$ and $C'$ below.
To obtain a labeling, we define $T$ and $T'$ to be spanning trees of the subgraphs $C$ and $C'$, respectively. 
To make $H$ rewirable, we fix $A_1 = V(C)$ and $A_2 = V(C')$, i.e., the important edges will connect the cut $(C,C')$ when being rewired.
We also fix the sets $S_1 = V(C)$ and $S_2 = V(C')$ as required by Def.~\ref{def:wc_graph}. 

First, we show Case~(i) which is for $\AllDetect$. Let $t= \frac{n}{\alpha}$. We define $C$ to be of size $t$ and $C'$ to be of size $(\alpha-1) t$. For the sake of readability, we assume that $t$ and $(\alpha-1)t$ are integers. Note that this assumption does not affect the asymptotic bounds. 
We equip nodes with the approximate network size $\tilde{n}$ = $t$. 
It is straightforward to verify that for $\alpha>1$, $\tilde{n}$ is indeed an $\alpha$-approximation to $|H|$ and $|S_1|$ as required by Def.~\ref{def:wc_graph}. 
Furthermore, since $\alpha$ is a constant, it follows that the rewirable components $A_1$ and $A_2$ each have size $\Omega(n)$. 
It remains to show that Property~(A), (B) and (C) of Def.~\ref{def:wc_graph} are satisfied.
Property~(A) holds  since $C$ and $C'$ are (disjoint) components. 
Thus, we only need to focus on Properties~(B) and (C). 
For Property~(B), consider the input $\Inp(H[S_1],L[S_1],\tilde{n}) = \Inp(C,T,t)$ which is legal since $T$ is a spanning tree for $C$. 
For Property~(C), consider any important edges $e_1$ and $e_2$. 
Graph $H^{e_1,e_2}$ is connected, but the subgraph induced by the labeling $L := T \cup T'$ is not. 
Thus the input $\Inp(H^{e_1,e_2},T \cup T', t)$ represents an illegal labeling.
This completes the proof of Case (i).

Next, we show Case~(ii) which is for $\OneDetects$. Let $t = \frac{n}{2}$.  We define $C$ and $C'$ to be of size $t$ each. Again, we assume $t$ is an integer. It follows that the rewirable components $A_1$ and $A_2$ each have size $\Omega(n)$. 
We equip nodes with the approximate network size $\tilde{n} = \alpha t$.  
It is trivial to see that for $\alpha \ge 1$, $\tilde{n}: = \alpha t$ is an $\alpha$-approximation to $t$ which is the size of $S_1$ and $S_2$.  
However, for $\tilde{n}: = \alpha t$ to be an $\alpha$-approximation to $|H|$ which has value $2t$, we require that $\tilde{n} \in [n/\alpha,\alpha\,n]$, which holds if and only if 
\[
\frac{2t}{\alpha} \le \tilde{n} = \alpha\cdot t \le 2\alpha\, t.
\]
In particular, the left inequality is true for any $\alpha \ge \sqrt{2}$, as required. 
We can show that Property~(A), (B) and (C) hold using the same argument as for Case (i),   
\end{proof}

\section{A Lower Bound for $d$-Approximate BFS Verification in $\mathsf{KT}_\rho$ ($\rho \ge 1$)} 
\label{sec:ktrho}

In this section, we consider $\ktx$, where $\rho \ge 1$. 
In $\mathsf{KT}_{\rho}$, each node $x$ is provided the initial knowledge of the IDs of all nodes at distance at most $\rho$ from $x$
and the neighborhood of every vertex at distance at most $\rho-1$ from $x$. 
For instance, when considering the special case $\ktone$, a node simply knows the IDs of all its neighbors. 
Formally, we define  $S_\rho(x) = \cup_{0\leq r \leq \rho} C_r(x)$, where $C_r(x)$ consists of all nodes that are at distance $r$ from $x$. 
Let $E_\rho(x) = \{(y,z) \in E \mid y, z \in S_\rho(x)\} \setminus \{(y,z) \in E \mid y, z \in C_\rho(x)\}.$ 
The \emph{$\rho$-neighborhood} of $x$ is the subgraph $N_\rho(x) = (S_\rho(x), E_\rho(x))$.
 
The indistinguishability strategy used for deriving the general $\ktzero$ lower bound (Lemma~\ref{lem:general_kt0_lb}) does not hold for $\ktx$. 
To begin with, one of the essential properties of our lower bound argument for $\ktzero$ is that a rewirable base graph has an identical initial state with any of its rewired graph, does not hold in $\ktx$.
This is because each node initially knows the IDs of its neighbors, and by construction, a node incident to important edges will have different neighbors in the rewired graph compared to the base graph. 

To mitigate this issue, we restrict ourselves to comparison-based algorithms, which operate under the following restrictions: 
We assume that each node has two types of variables: \emph{ID variables} and \emph{ordinary variables}. 
Initially, ID variables contain the IDs known to the node before the start of the execution of a protocol. 
The ordinary variables initially contain some constants known to the node, for example, an approximation of the network size $\tilde{n}$. 
The \emph{state} of a node $u$ consists of the combined lists of ID and ordinary variables.

The local computations of a node is limited to comparing ID variables and storing the result of the comparison in an ordinary variable,
as well as performing some arbitrary computations on its ordinary variables. 
Essentially, the specific values of the IDs learned by a node do not influence the outcome of the comparison-based distributed algorithm; instead, the relative ordering of the IDs is what matters.
Note that these are standard assumptions for comparison-based algorithms (see  e.g., \cite{awerbuch1990trade,frederickson1987electing,pai2021can}). We need the following definition.  
\medskip
\begin{definition}[order-equivalent]
Consider a graph $H$ and two ID assignments  $\phi$ and $\psi$ of $V(H)$. 
We say that $\phi$ and $\psi$ are \emph{order-equivalent} if for any two distinct nodes $x$ and $y$ of $V(H)$\footnote{Recall that our model assumes that each node has a unique ID.}, the following holds: 
$\phi(x) < \phi(y)$ if and only if $\psi(x) < \psi(y)$. 
Let $G$ be an isomorphic graph of $H$, ignoring ID assignments. 
We say that $G$ and $H$ are order-equivalent if the ID assignments of $G$ and $H$ are order-equivalent. 
\end{definition}

Recall that in $\ktx$, the initial state of a node is defined by its $\rho$-neighborhood. 
Hence, when we consider a deterministic comparison-based algorithm, if the $\rho$-neighborhood of two nodes are order-equivalent, then they will behave the same initially. 
This observation allows us to generalize the indistinguishability strategy used for $\ktzero$ to $\ktx$. 

Throughout this section, we consider comparison-based algorithms in the $\mathsf{KT}_{\rho}$, $\congest$ model, where each message sent by a node may include $B=O(1)$ node IDs. Apart from the result for broadcast of \cite{awerbuch1990trade}, we are not aware of any other superlinear (in $n$) message complexity lower bounds in the $\kt_\rho$ setting that hold for $\rho\ge2$.\footnote{We point out that \cite{awerbuch1990trade} only provides a full proof of their result for the special case $\rho=1$.}

In the following, we provide an informal discussion of our lower bound approach for $\ktx$: 
Similar to the indistinguishability strategy for $\ktzero$, we consider a base graph and a rewired variant of the base graph. 
We argue that the executions of any comparison-based algorithm on both graphs are indistinguishable to the nodes, as long as the important edges are not ``utilized". In $\ktzero$, an edge is ``utilized" if it is used to send or receive messages. However, in $\ktx$, this notion of utilization is insufficient.

In \cite{awerbuch1990trade}, the concept of ``connectedness" is introduced to define the utilization of an edge, while the ``charging rule" is presented to bound the message complexity based on the number of utilized edges. However,  the argument in the lower bound approach of \cite{awerbuch1990trade} crucially exploits the assumption that their base graph is disconnected. For showing a lower bound for the $d$-approximate BFS tree verification problem, we need the base graph as well as the rewired graphs to be connected. Thus, we need to introduce different rules for \emph{utilizing} and \emph{charging} edges, which have the added benefit of leading to a significantly simplified indistinguishability proof, as we will see in the proof of Lemma~\ref{lem:ktx_indistinguishable} below.

Even though it is relatively straightforward to generalize our approach to other verification problems such as spanning tree verification, here we exclusively focus on verifying a $d$-approximate BFS tree, which is more challenging in terms of techniques, as it requires us to deal with \emph{connected} graphs. As mentioned above, the existing strategy of \cite{awerbuch1990trade} does not apply to connected base graphs.
In the remainder of this section, we prove Theorem~\ref{thm:ktx_bfs_lb}.

\begin{theorem}  \label{thm:ktx_bfs_lb}
Let $0 < \beta < 1$ and $0 < \epsilon < \frac{1}{4}$ be constants. 
Consider any $\rho\ge1$, and $\frac{3\rho-1}{\rho+1} \le d \le \frac{(1-\beta)n}{4\rho-2}$.
For any $\epsilon$-error randomized comparison-based algorithm that solves $d$-approximate BFS tree verification under the $\mathsf{KT}_{\rho}$ assumption with the guarantee that at least one node detects illegal inputs (i.e., $\OneDetects$), there exists a network where its message complexity is $\Omega\left(\frac{1}{\rho}\left(\frac{n}{\rho}\right)^{1+\frac{c}{\rho}}\right)$, for some constant $c >0$. 
This holds even when all nodes know the exact network size.
For the special case $\rho=1$, we obtain a lower bound of $\Omega\lt( n^2 \rt)$.
\end{theorem}

We remark that our message complexity lower bounds are determined by the number of edges of the rewirable components. 
For $\rho=1$, the rewirable components are cliques with $\Omega(n^2)$ edges. 
For $\rho>1$, the rewirable components are graphs with girth at least $2\rho$ and with $\Omega\lt(\lt({n}/{\rho}\rt)^{1+\frac{c}{\rho}}\rt)$ edges. 
When $\rho$ is a constant, this results in a lower bound on message complexity of the form $n^{1+\Omega\lt( 1 \rt)}$. However, as $\rho$ approaches $\log(n)$, the lower bound deteriorates to the trivial bound of $\Omega(n)$ on the message complexity.

\paragraph*{The Lower Bound Graph Family}
For a given $\rho \ge 1$ and $d \geq \frac{3\rho-1}{\rho+1}$, we construct an infinite family of $n$-node graphs (and their respective rewired variants) that yield the claimed bound of Theorem~\ref{thm:ktx_bfs_lb}. 
We consider a graph $H$ of $n$ vertices that resembles the lower bound graph that sufficed for the $\kt_0$ assumption (see Lemma~\ref{lem:bfs_graph}).
In the $\kt_\rho$ setting, a crucial difference is that, in order to show that 
the initial state of a node $x$ has an order-equivalent $\rho$-neighborhood in $H$ and in $H^{e,\tilde{e}}$, we need to add $\rho-1$ layers before and after each of the rewirable components $A_1$ and $A_2$. 

Figure~\ref{fig:ktrho} on page~\pageref{fig:ktrho} gives an example of the graph construction that we now describe in detail:
The vertices of $H$ are partitioned into $(\rho+1)d+\rho+1$ \emph{levels} numbered $0,\dots,(\rho+1)d+\rho$. 
Let $N = [1,2\rho-1]$  and $N' = [(\rho+1)d-\rho+2,(\rho+1)d+\rho]$.
The vertices of levels in $N$ correspond to the nodes of the first $2\rho-1$ layers after the layer $0$. 
The vertices of levels in $N'$ correspond to the nodes of the last $2\rho-1$ layers.
Let $\gamma = k n$ where the value of $k$ is specified later.\footnote{For the sake of readability, we assume that $\gamma$ is an integer, which does not affect the asymptotic bounds.}
For each $i \in N \cup  N'$, level $i$ contains $\gamma$ nodes each, denoted by $u_1^{(i)},\ldots,u_{\gamma}^{(i)}.$
All the other levels consist of only a single node, and we use $u^{(i)}$ to denote the (single) node on level $i$ for all $i \in [0,(\rho+1)d+\rho ] \setminus  \{N \cup  N'\}.$ 

Next, we define the edges of $H$:
There is an edge from the root to every node in level $1$. 
For each $i \in N \setminus \{1\} \cup  N' \setminus \{(\rho+1)d-\rho+2\}$, 
there is an edge from nodes $u^{(i-1)}_j$ to $u^{(i)}_j$, for $1 \leq j\leq \gamma$. 
For each $i \in [2\rho,(\rho+1)d-\rho+2]$, all nodes at level $i-1$ have an edge to every node at level $i$; note that a level may only contain one such node. 
The nodes at level $\rho$ form a subgraph $C_\rho$ (defined below) and, similarly, the nodes at level $d(\rho+1)+1$ form another subgraph $C_\rho'$ which is isomorphic to $C_\rho$, when ignoring node IDs. 
We call the edges that connect nodes at different levels the \emph{inter-level edges}. 

For a node $x$ of level $i$ in $N$ (resp.\ $x$ of level $i$ in $N'$), we use the notation $\tilde{x}$ to refer to its copy in level $i$ of $N'$ (resp.\ in level $i$ of $N$), and we extend this notation to the edges between the nodes in levels $N$ and between nodes in levels of $N'$ in the natural way.  
An edge $e = (u,v) \in C_\rho$ induces the \emph{rewired graph} $H^{e,\tilde{e}}$, where $\tilde{e} = (\tilde{u}, \tilde{v})$ is $e$'s copy in $C_\rho'$.

Let $T$ be the directed tree, which contains all nodes of $H$ and all the directed inter-level edges from level $i$ to $i+1$ for all $i$ except for $i=2\rho-1$, in which case, only the edge $(u_1^{(2\rho-1)}, u^{(2\rho)})$ is in $T$.   
Then, $T$ is indeed a ($1$-approximate) BFS tree of $H$.
However, $T$ is not a $d$-approximate BFS tree of $H^{e,\tilde{e}}$ because  
$\dist_{T}(u^{(0)}, u_k^{({d(\rho+1)+1})}) = d(\rho+1)+1$ for $k\geq 1$, whereas $\dist_{H^{e,\tilde{e}}}(u^{(0)}, u_i^{({d(\rho+1)+1})}) = \rho+1$ for $u_i^{({d(\rho+1)+1})}$ which is incident to $e$ or $\tilde{e}$. 
Hence the algorithm must produce different outputs on these two graphs when the input labeling is $T$.

\paragraph*{The subgraph $C_\rho$}
Intuitively speaking, the achieved message complexity lower bound will be determined by the number of edges in $C_\rho$, and thus we aim to make $C_\rho$ as dense as possible.
For $\rho = 1$, we simply define $C_1$ (and hence also $C_1'$) to be the clique on $\gamma$ vertices. 

To simplify our argument needed for the lower bound result, we require that $C_\rho$ satisfies the \emph{$\rho$-unique shortest path} property:
\medskip
\begin{definition}[$\rho$-unique shortest path] \label{def:unique_shortest_path}
A graph is said to satisfy the \emph{$\rho$-unique shortest path} property if for any two nodes that are of distance at most $\rho$ from each other, there exists a unique path of that length between the two nodes. 
\end{definition}
\medskip

This trivially holds for $\rho=1$ since $C_1$ and $C_1'$ are cliques. 
For $\rho \ge 2$, we define $C_\rho$ (and hence also $C_\rho'$) to be a subgraph of $\gamma$ nodes with girth greater than $2\rho$ and with $\Omega(\gamma^{1+\frac{c}{\rho}})$ edges for a constant $c > 0$. Recall that the \emph{girth} is the length of the shortest cycle in a graph and that a graph with a girth greater than $2\rho$ has the $\rho$-unique shortest path property that we specified earlier. 
The existence of such graphs $C_\rho$ is known, see \cite{bollobas2004extremal}.
 
We summarize the properties of $H$ in the following lemma: 
\medskip
\begin{lemma} \label{lem:ktrho_graph} 
The graph $H$ contains subgraphs $C_\rho$ and $C_\rho'$ which satisfy the $\rho$-unique shortest path property, and $|V(C_\rho)| = |V(C_\rho')| = \gamma$.
If $\rho=1$, the subgraph $C_1 \cup C_1'$ has $\Omega(n^2)$ edges.
If $\rho\ge2$, the subgraph $C_\rho\cup C_\rho'$ has $\Omega(\gamma^{1+\frac{c}{\rho}})$ edges for some constant $c>0$. %
The directed tree $T$, which contains all nodes of $H$ and all directed inter-level edges from level $i$ to $i+1$ for all $i$ except for $i=2\rho-1$, in which case, only the edge $(u_1^{(2\rho-1)}, u^{(2\rho)})$ is in $T$,  is a $d$-approximate BFS tree of $H$ but not of any rewired graph $H^{e,\tilde{e}}$. 
\end{lemma}
\begin{figure}[th]
  \centering
\begin{subfigure}[t]{0.42\textwidth}  
  \centering
	\includegraphics[scale=0.75]{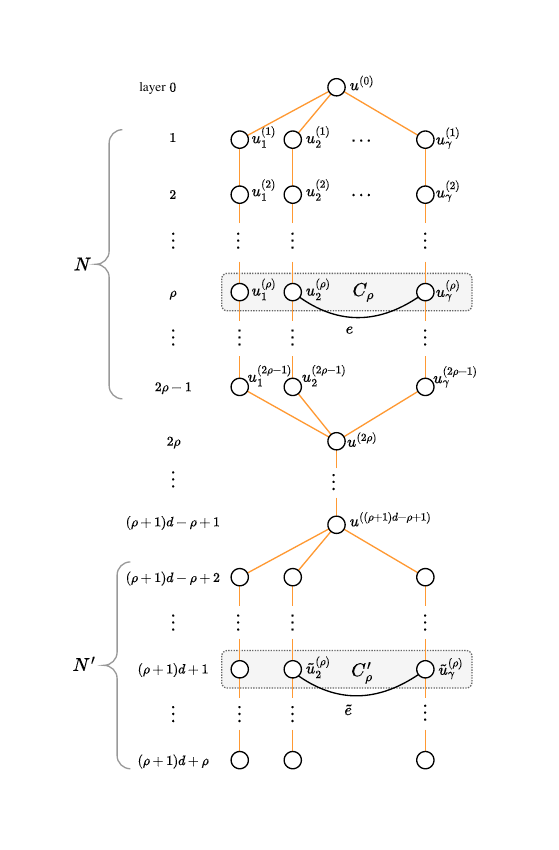}
  \caption{
		The base graph $H$. 
  }
  \label{fig:ktrho1}
\end{subfigure}
\hfill
\begin{subfigure}[t]{0.42\textwidth}
  \centering 
	\raisebox{5pt}{\includegraphics[scale=0.75]{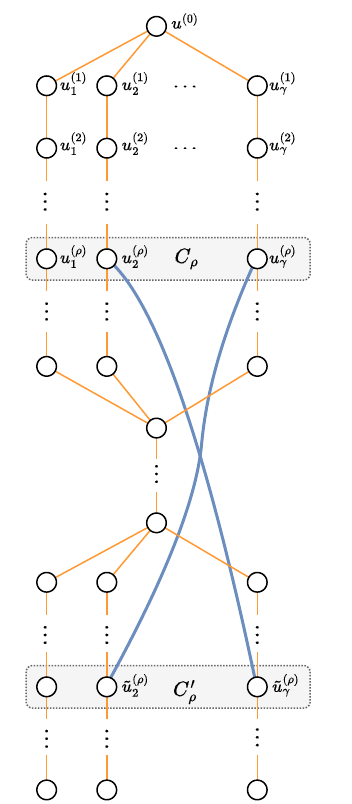}}
  \caption{
		The rewired graph $H^{e,\tilde e}$. 
  }
  \label{fig:ktrho2}
\end{subfigure}
\caption{The graph construction for proving a lower bound in $\kt_\rho$ for verifying a $d$-approximate BFS tree, where we have omitted the input labeling and node IDs. The shaded subgraphs are high girth graphs on $\gamma$ nodes that have $\Omega\lt( \gamma^{1+\frac{c}{\rho}} \rt)$ edges. We emphasize that even the base graph $H$ is connected, and thus the approach of \cite{awerbuch1990trade} does not apply in our setting.}
\label{fig:ktrho}
\end{figure}
\medskip
Next, we present a suitable ID assignment $\phi$ for vertices of $H$ (which also applies to vertices of $H^{e,\tilde{e}}$) defined as follows: 
\begin{enumerate}
    \item assign a distinct even integer in $[0, 2(2\rho-1)\gamma]$ to each node belonging to a level in $N$;
    \item assign a distinct odd integer in $[1,2(2\rho-1)\gamma+1]$ to each node $x$ belonging to a level in $N'$ such that $\phi(\tilde{x}) = \phi(x)+1$;
    \item assign arbitrary unique integers from $[2(2\rho-1)\gamma+2, n]$ to the remaining nodes in $V(H)$.
\end{enumerate} 
\medskip

\begin{definition}[$e$-connected]
Consider an edge $e$ in a graph that satisfies the $\rho$-unique shortest path property.  We say that nodes $x$ and $z$ are \emph{$e$-connected}, denoted by $x\econn z$, if $\dist(x, z) \leq \rho$ and the (unique) shortest path from $x$ to $z$ contains $e$. 
\end{definition}

\medskip
\begin{definition}[utilized/unutilized edge] \label{def:utilized}
We say that an edge $e \!=\! (u, v)$ is \emph{utilized} if one of the following conditions hold: 
\begin{compactenum}
    \item a message is sent across $e=(u,v)$; 
    \item there are two nodes $x$ and $z$ such that $x\econn z$, and $x$ receives or sends the ID of $z$ or $\tilde{z}$. 
\end{compactenum}
Otherwise, we say that the edge $e$ is \emph{unutilized}.
\end{definition} 

\medskip
\begin{definition}[charging rule] \label{def:charging}
If an edge $(x,y)$ is used to send a message that contains the ID of node $z$, we charge one unit to each of the following edges: 
\begin{compactenum}
    \item $(x,y)$;
    \item for each edge $e$, such that either $x\econn z$ or $x\econn \tilde{z}$.
    \item for each edge $e$, such that either $y\econn z$ or $y\econn \tilde{z}$.
\end{compactenum}
\end{definition}

\medskip
\begin{lemma}\label{lem:ktrho_lb}
Consider a graph $H$ which satisfies the $\rho$-unique shortest path property. 
Let $\mu$ be the number of utilized edges in an execution on $H$. Then the message complexity of the execution is $\Omega(\mu/\rho)$. 
\end{lemma}

\begin{proof}
Let $\mu$ be the number of utilized edges, $M$ be the number of messages sent and $C$ be the total cost charged in the execution.
A message sent on an edge $(x,y)$ will incur one unit of charge to $(x,y)$ and for each $z$ where its ID is in the message, a unit of charge is applied to each edge $e$ where $x\econn z$,  $x\econn \tilde{z}$, $y\econn z$ and $y\econn \tilde{z}$. 
Note that there are at most $\rho$ edges of each of these types since the graph $H$ satisfies the $\rho$-unique shortest path property.
Note that each message can carry at most $B$ IDs of other nodes. 
Consequently, for each of the messages sent, at most $1+4B\rho$ edges are charged.
Hence, we have $C \leq (1+4B\rho)M.$
On the other hand, each utilized edge is charged at least once. Hence,
$\mu \leq C.$
As a result, we have $M \geq \frac{\mu}{1+4B\rho}$ and the lower bound $\Omega(\mu/\rho)$ follows.
\end{proof}

\begin{lemma}\label{lem:local_neighborhood}
Let $e = (u,v)$ be an edge in $C_\rho$. For $x \in \{u, v, \tilde{u}, \tilde{v}\}$, we have  $N_{\rho-1}(x)$ in $H$ is order-equivalent to $N_{\rho-1}(\tilde{x})$ in $H^{e, \tilde{e}}$. Consequently, $N_{\rho'}(x)$ in $H$ is order-equivalent to $N_{\rho'}(\tilde{x})$ in $H^{e, \tilde{e}}$ for any $\rho'\le \rho-1$.
\end{lemma}
\begin{proof}
Let $e'$ and $\tilde{e}'$ be the ``rewired" edges of $e$ and $\tilde{e}$ in $H^{e, \tilde{e}}$, respectively. Formally, $e' = (u, \tilde{v})$ and $\tilde{e}' = (\tilde{u}, v)$.  Note that $H$ is constructed to satisfy the following: 
Nodes $w$, $z$ are $e$ (resp. $\tilde{e}$)-connected in $H$ if and only if $w$, $\tilde{z}$ are $e'$(resp. $\tilde{e}'$)-connected in $H^{e, \tilde{e}}$.
Hence, $N_{\rho-1}(x)$ in $H$ is isomorphic (ignoring the ID assignments) to $N_{\rho-1}(\tilde{x})$ in $H^{e, \tilde{e}}$. 
Next, we need to show that the ID assignments of $N_{\rho-1}(x)$ is order-equivalent in both graphs. 
We abuse notation and say that node $x$ is in $N$ (resp. $x \in N'$) if $x$ is a node at a level belonging to $N$ (resp. $N'$). 
The definition of $\phi$ asserts that for two distinct nodes $x, y \in N$ (or $x, y \in N')$, we have $\phi(x) < \phi(y)$ if and only if $\phi(\tilde{x}) < \phi(\tilde{y})$ and $\phi(x) < \phi(\tilde{y})$. This establishes the order-equivalence property we require. 
\end{proof}

We are now ready to present the indistinguishability result. 
For now, we focus on deterministic algorithms; we later extend the result to randomized algorithms via a simple application of Yao's lemma.

\begin{lemma} \label{lem:ktx_indistinguishable}
Consider a deterministic comparison-based algorithm $\mathcal{A}$, the ID assignment $\phi$, graph $H$ and any rewired graph $H^{e,\tilde{e}}$, where $e \in C_\rho$. If $e$ and $\tilde{e}$ are both unutilized (see Def.~\ref{def:utilized}), then $H$ and $H^{e,\tilde{e}}$ are indistinguishable for every node $u$ when executing $\mathcal{A}$, i.e., $u$ has an order-equivalent $\rho$-neighborhood in both networks, and it sends and receives the same sequence of messages in $H$ as it does in $H^{e,\tilde{e}}$.  
\end{lemma}
\begin{proof}
Let $e = (u,v)$. 
The statement is immediate for any node not in the $(\rho-1)$-neighborhood of either $u$, $v$, $\tilde{u}$, or $\tilde{v}$. 
Thus we focus on these nodes. We denote the set of these nodes as $N^* = N_{\rho-1}(u) \cup N_{\rho-1}(v) \cup N_{\rho-1}(\tilde{u}) \cup N_{\rho-1}(\tilde{v})$. 

First, we show the following claim.  
\begin{claim}
A node $x \in N^*$ has an order-equivalent $\rho$-neighborhood in $H$ and in $H^{e, \tilde{e}}$, i.e., $N_\rho(x)$ in $H$ is order-equivalent to $N_\rho(x)$ in $H^{e, \tilde{e}}$.
\end{claim}
\begin{proof}
For a node $x$ and an edge $f$, we define $S_{x, f}$ to be the set of nodes $x'$ in $N_{\rho-1}(x)$ such that the shortest path from $x'$ to $x$ does not include edge $f$ (i.e., $x'$ and $x$ are not $f$-connected).
Then, $N^*$ is partitioned into disjoint subsets of 
$S_{u,e}$, $S_{v, e}$, $S_{\tilde{u}, \tilde{e}}$ and $S_{\tilde{v}, \tilde{e}}$. 
We first show that the claim is true for $x=u$. 
Note that $N_\rho(x)$ in $H$ is order-equivalent to $N_\rho(x)$ in $H^{e, \tilde{e}}$ if and only if $N_{\rho-1}(w)$ in $H$ is order-equivalent to $N_{\rho-1}(w)$ in $H^{e, \tilde{e}}$ for all neighbors $w$ of $x$. 
Observe that the set of neighbors of \( u \) in \( H \) are the same as in \( H^{e, \tilde{e}} \), except that \( u \) is adjacent to \( v \) in \( H \) and to \( \tilde{v} \) in \( H^{e, \tilde{e}} \). 
Hence, the claim that $N_{\rho}(u)$ in $H$ is order-equivalent to $N_{\rho}(u)$ in $H^{e, \tilde{e}}$ is true if and only if $N_{\rho-1}(v)$ in $H$ is order-equivalent to $N_{\rho-1}(\tilde{v})$ in $H^{e, \tilde{e}}$, which we have established in Lemma~\ref{lem:local_neighborhood}.
The above argument for the special case $x=u$ can be generalized to all $x\in S_{u, e}$ as follows: 
The claim that $N_{\rho}(x)$ in $H$ is order-equivalent to $N_{\rho}(x)$ in $H^{e, \tilde{e}}$ is true if and only if $N_{\rho-r-1}(v)$ in $H$ is order-equivalent to $N_{\rho-r-1}(\tilde{v})$ in $H^{e, \tilde{e}}$ where $r = \dist_H(x, u)$, which we have established in Lemma~\ref{lem:local_neighborhood}. 
A similar argument applies to the nodes in $S_{v, e}$, $S_{\tilde{u}, \tilde{e}}$ and $S_{\tilde{v}, \tilde{e}}$. 
\end{proof}

Next, we show that each node sends the same messages in round $1$. 
Since $x$ has order-equivalent neighborhoods in both executions and recalling that the algorithm is comparison-based, any ordinary variable that $x$ computes at the start of round $1$ must have the same value in both executions.
Next, we show that the ID type variable used for $x$ in sending messages is identical in both executions as well. 
Recall that the IDs that $x$ knows at the start of round $1$ are exactly the IDs in its $\rho$-neighborhood, which may be different in $H$ compared to $H^{e,\tilde{e}}$, as we have observed above. 
In particular, the IDs in $N_\rho(x)$ of $H$ that are not in $N_\rho(x)$ of  $H^{e,\tilde{e}}$ are the IDs of all vertices $z$, for which it holds that $x$ and $z$ are $e$-connected. 
The reason for this is that, in $H^{e,\tilde{e}}$, nodes $x$ and $\tilde{z}$ are $e$-connected, whereas $x$ and $z$ are not.  
However, since $e$ and $\tilde{e}$ are unutilized, $x$ can neither include the ID of $z$ in the execution on $H$, nor the ID of its neighbor $\tilde{z}$ in the execution on $H^{e,\tilde{e}}$, for any node $z$ for which $x\econn z$.
Hence, all IDs in the message sent by node $x$ are identical in both executions. 
Therefore, $x$ sends the same messages in both networks, $H$ and $H^{e,\tilde{e}}$.

Finally, we show that if every node sends the same messages during round $r-1$, then every node sends the same messages in round $r$.  
We use induction over the rounds, where the basis already follows from above.
Now consider some round $r>1$ and assume that every node sent the same messages in round $r-1$ in both $H$ and $H^{e, \tilde{e}}$. 
In this case, every node also receives the same messages during round $r-1$ in both executions.
Consider the set of messages $\Pi$ received by $x$.
Since $e$ and $\tilde{e}$ are unutilized, no message in $\Pi$ contains the ID of $z$ or $\tilde{z}$ for all nodes $z$ that $x$ is $e$-connected to.
Hence, all IDs received by $x$ correspond to the same nodes in $H$ and $H^{e,\tilde{e}}$. 
Again, using the fact that the algorithm is comparison-based and the fact that all IDs that may be contained in the messages in $\Pi$ belong to nodes that have an identical $\rho$-neighborhood in $H$ and $H^{e,\tilde{e}}$, it follows that $x$ sends the same messages in round $r$.
\end{proof}

We are now ready to complete the proof of Theorem \ref{thm:ktx_bfs_lb}
for deterministic comparison-based algorithm. 

\begin{lemma} \label{lem:kthro_deterministic}
Consider any small constant $\eta>0$.
If a deterministic comparison-based algorithm solves $d$-approximate BFS verification on the base graph $H$ and for at least a $\eta$-fraction of the rewired graphs, then it has a message complexity of $\Omega\left(\frac{1}{\rho}\left(\frac{n}{\rho}\right)^{1+\frac{c}{\rho}}\right)$ for $\rho>1$ (resp. $\Omega(n^2)$ for $\rho=1$).
\end{lemma}
\begin{proof} 
Assume towards a contradiction that there exists a deterministic comparison-based algorithm $\A$ that solves $d$-approximate BFS verification on input $\Inp(H,T,n)$ with message complexity $o\left(\frac{1}{\rho}\gamma^{1+\frac{c}{\rho}}\right)$ for $\rho>1$ (resp. $o(n^2)$ for $\rho=1$). 
Then by Lemma \ref{lem:ktrho_lb}, the number of utilized edges is  $o(\gamma^{1+\frac{c}{\rho}})$ for $\rho>1$ (resp. $o(n^2)$ for $\rho=1$).
Lemma \ref{lem:ktrho_graph} tells us that there are $\Omega (\gamma^{1+\frac{c}{\rho}})$ edges in $C_\rho \cup C_\rho'$ for $\rho>1$ (resp. $\Omega(n^2)$ for $\rho=1$). 
Hence, there is a subset $\mathcal{E}$ containing at least a $(1-\eta)$-fraction of the edges in $C_\rho$ such that, for every $e \in \mathcal{E}$, both $e$ and $\tilde{e}$ are unutilized. 
Consider the rewired graph $H^{e,\tilde{e}}$. 
Lemma~\ref{lem:ktx_indistinguishable} ensures that every node outputs the same result in both executions. However,  according to Lemma \ref{lem:ktrho_graph}, $T$ is a $d$-approximate BFS of $H$ but not of $H^{e,\tilde{e}}$, which shows that the algorithm fails on every rewired graph $H^{e,\tilde{e}}$, for $e \in \mathcal{E}$. 

The claimed bound on the message complexity for $\rho=1$ follows immediately from the above discussion.
For $\rho>1$, we need to show that  $\Omega\left(\frac{1}{\rho}\gamma^{1+\frac{c}{\rho}}\right)=\Omega\left(\frac{1}{\rho}\left(\frac{n}{\rho}\right)^{1+\frac{c}{\rho}}\right)$.
Let $\gamma = kn$ where 
$k = \frac{1}{4\rho-2} - \frac{1}{n}\cdot\frac{d(\rho+1)-3\rho+3}{4\rho-2}$. 
We have 
\begin{align}
  \gamma
    &=  \frac{n}{4\rho-2} - \frac{d(\rho+1)}{4\rho-2} + \frac{3(\rho-1)}{2(2\rho-1)}\notag\\
    \ann{since $\rho \ge 1$, $\frac{\rho-1}{2\rho-1}\ge 0$}
    &\ge \frac{n}{4\rho-2} - \frac{d(\rho+1)}{4\rho-2} \notag\\
    \ann{since $d \leq \frac{(1-\beta)n}{\rho+1}$}
    &\ge \frac{\beta n}{4\rho-2} = \Omega\left(\frac{n}{\rho}\right). \notag
  \end{align}
\end{proof}

\noindent\textbf{Randomized Algorithms.}
So far, we have restricted our attention to deterministic algorithms. 
To complete the proof of Theorem \ref{thm:ktx_bfs_lb}, we need to extend this result to randomized Monte Carlo algorithms that fail with some small probability $\epsilon < 1/4$.
We follow the standard approach of showing a lower bound for deterministic algorithms that succeed with a sufficiently large probability, when sampling the input graph from a hard distribution, defined next: 
We first flip a fair coin that determines whether we choose the base graph $H$ or a rewired graph. In the latter case, we sample a rewired graph uniformly at random from the set of all possible rewired graphs $\mathcal{R}$.

\medskip

Now consider a deterministic algorithm $\mathcal{A}$ that succeeds on this distribution with a distributional error of at most $2\epsilon$.
Observe that the given algorithm $\mathcal{A}$ cannot fail on graph $H$, since this would result in an distributional error of at least $\tfrac{1}{2}> 2\epsilon$.
Thus, we will obtain a contradiction if we can show that $\mathcal{A}$ fails on more than a $4\epsilon$-fraction of the rewired graphs, as this would yield a distributional error of more than $\tfrac{4\epsilon}{2}=2\epsilon$.%
If $\mathcal{A}$ does not satisfy the sought message complexity lower bound, then Lemma~\ref{lem:kthro_deterministic} tells us that $\mathcal{A}$ fails on at least a $(1-\eta)$-fraction of the rewired graphs. 
By choosing $\eta$ sufficiently small, it follows that $1 - \eta > 4\epsilon$ if $\epsilon<\frac{1}{4}-\frac{\eta}{4}$, where the latter term can be made as small as needed.
Thus, similarly to the proof of Lemma~\ref{lem:general_kt0_lb} in Section~\ref{app:general_kt0_lb}, a simple application of Lemma~\ref{lem:yao} yields the sought message complexity lower bound for randomized Monte Carlo algorithms, and completes the proof of the theorem.

\section{An Algorithm for Spanning Tree Verification in the $\ktzero$ $\congest$ Model} 
\label{sec:st_algo}\label{sec:st_upper}   

In this section, we give a message-efficient algorithm that verifies whether the input $T$ is an ST of the network $G$ under the $\OneDetects$ assumption, in the setting where all nodes have knowledge of some $\alpha$-approximation $\tilde{n}$ of the network size $n$, for some $\alpha < \sqrt{2}$; formally speaking, $\tilde{n} \in [n/\alpha,\alpha n]$.
Our result stands in contrast to the strong lower bound of $\Omega(n^2)$ shown by \cite{kor2013tight} that holds for deterministic ST verification assuming $\AllDetect$ and without \emph{any} knowledge of the network size.
Note that the bound of $\alpha<\sqrt{2}$ is tight, since  we show a lower bound of $\Omega(n^2)$ under the same setting in Theorem~\ref{thm:bfs_kt0_lb} in Section~\ref{sec:general_lb}. 

We obtain our algorithm by adapting the classic GHS algorithm for constructing an ST, see \cite{gallager1983distributed}.
However, in contrast to the GHS algorithm, we do not employ the communication-costly operation of exchanging the fragment IDs between neighboring nodes, which requires $\Omega\lt( m \rt)$ messages per iteration.
Considering that our goal is to verify that a given tree $T$ is indeed a spanning tree, we can select an arbitrary edge $e$ from $T$ (when growing a fragment) that is incident to some vertex of the fragment: we can stop the growing process immediately if $e$ turns out to close a cycle. 

\begin{theorem} \label{thm:mst_upper}
Suppose that all nodes know an $\alpha$-approximation of the network size, for some $\alpha<\sqrt{2}$. 
There is a deterministic $\ktzero$ $\congest$ algorithm that solves spanning tree verification with a message complexity of $O(n \log n)$ and a time complexity of $O(n\log n)$ rounds while ensuring at least one node detects illegal inputs (i.e., $\OneDetects$).
Moreover, if nodes have perfect knowledge of the network size, the algorithm guarantees $\AllDetect$.
\end{theorem}

While our main focus is on the message complexity, we point out that the round complexity of the algorithm in Theorem~\ref{thm:mst_upper} can be as large as $O(n\log n)$. 
This comes as no surprise, considering that the state-of-the-art solution~\cite{DBLP:conf/icdcn/MashreghiK17,king2015construction} for computing a spanning tree with $O(n \poly\log n)$ messages (even under the stronger $\ktone$ assumption) requires at least $\Omega(n)$ rounds. 

In Section~\ref{subsec: ST_algo}, we describe and analyze the deterministic algorithm claimed in Theorem~\ref{thm:mst_upper}. 
\onlyShort{We have omitted some proofs, which can be found in the full paper.}

\subsection{Description of the Algorithm}\label{subsec: ST_algo}

\paragraph*{Growing Fragments:} 
Each node forms the root of a directed tree, called \emph{fragment}, that initially consists only of itself as the \emph{fragment leader}, and every node in the fragment knows its (current) fragment ID, which is simply the ID of the fragment leader.

The algorithm consists of iterations each comprising $c_1\,\alpha\,\tilde n$ rounds \footnote{Recall that each node only knows the  $\alpha$-approximation network size $\tilde{n}$. By definition of $\alpha$-approximation, we have $\alpha \tilde{n} \ge n$.}, for a sufficiently large constant $c_1\ge 1$, and the goal of an iteration is to find an \emph{unexplored edge}, i.e., some edge $e$ over which no message has been sent so far. 
We point out that, in contrast to other ST construction algorithms that follow the Boruvka-style framework of growing fragments, $e$ is \emph{not} guaranteed to lead to another fragment. 

In more detail, we proceed as follows:
At the start of an iteration, each fragment leader broadcasts along the edges of its fragment $F$. 
Upon receiving this message from a parent in $F$, a node $u$ checks whether it has any incident edges in $T$ over which it has not yet sent a message.
If yes, $u$ immediately responds by sending its ID to its parent; otherwise, it forwards the request to all its children in $F$.
If $u$ does not have any children, it immediately sends a nil-response to its parent. 
On the other hand, if $u$ does have children and it eventually receives a non-nil message from some child, it immediately forwards this response to its parent and ignores all other response messages that it may receive from its other children in this iteration.
In the case where $u$ instead received nil responses from every one of its children, $u$ finally sends a nil message to its own parent.
This process ensures that the fragment leader $u_\ell$ eventually learns the ID of some node $v$ in its fragment that has an unexplored incident tree edge $e_v\in T$, if such a $v$ exists. 

Next, $u_\ell$ sends an \texttt{explore!}-message along the tree edges that is forwarded to $v$, causing $v$ to send a message including the fragment ID (i.e., $u_\ell$'s ID) on edge $e_v$, over which it has not yet sent a message. 
Assume that this edge is connected to some node $w$.

We distinguish two cases:
First, assuming that $w$ is in the same fragment as $v$, node $w$ responds by sending an $\texttt{illegal}$ message to $v$ who upcasts this message to the fragment leader $u_\ell$, who, in turn, initiates a downcast of this message to all nodes in the fragment, instructing every node in $F$ to output $0$ (``reject''). 
In this case, the fragment $F$ stops growing and all its nodes terminate.

In the second case, $w$ is in a distinct fragment. 
If $w$'s fragment has not yet terminated, then $w$ responds to $v$,
and $v$ relays the ID of the successfully-found outgoing fragment edge to its fragment leader. 
Moreover, $w$ informs its own fragment leader about the fragment ID of $v$.
If, on the other hand, $w$ has already terminated, it remains mute, and $v$ forwards this information to its own fragment leader, causing all nodes in $F$ to terminate.
Before we start the next iteration, all fragment leaders simultaneously start a cycle detection procedure described below, initiated exactly $c_1 \alpha \tilde n$ rounds after the start of current iteration.

\paragraph*{Acyclicity Check:}
Since all (non-terminated) fragments attempt to find outgoing edges in parallel in this iteration, we may arrive at the situation where there is a sequence of fragments $F_1,F_2,\dots,F_k$ such that the unexplored edge found by $F_i$ leads to $F_{i+1}$, for $i \in [1,k-1]$, and the unexplored edge discovered by $F_k$ may lead ``back'' to some $F_j$ ($j < k$).
Conceptually, we consider the \emph{fragment graph} $\mathcal{F}$ where vertex $f_i$ corresponds to the fragment leader of $F_i$, and there is a directed edge from $f_i$ to $f_j$ if the edge explored by $F_i$ points to some node in $F_j$. %
If two fragments $f$ and $f'$ happen to both explore the same edge $e$ in this iteration, then we say that $e$ is a \emph{core edge}.
We say that a cycle in ${\mathcal{F}}$ is \emph{bad} if it involves at least $3$ fragments.

\medskip
\begin{lemma} \label{lem:frags}
Every component $C$ of ${\mathcal{F}}$ has at most one bad cycle. 
Moreover, a component $C$ that contains at least two fragments does not have a bad cycle if and only if there exist exactly two fragments $f$ and $f'$ in $C$ that are connected by a core edge.
\end{lemma}
\begin{proof}
Since each fragment explores one outgoing edge, it is clear that $C$ has at most one bad cycle, which proves the first statement.

Now suppose that $C$ has no bad cycles. 
Recalling that each fragment in $C$ has exactly one outgoing edge in $\mathcal{F}$, it follows that there must exist a core edge between two fragments, as otherwise there would be a bad cycle. 
Since $C$ has the same number of edges as it has fragments and is connected, there cannot be any other core edges.

For the converse statement, suppose that $C$ has exactly two fragments $f$ and $f'$ connected by a core edge $e$, and assume towards a contradiction that there exists a bad cycle $Z$ in $C$ consisting of $k$ edges, for some positive integer $k\ge 3$. 
Note that $e$ cannot be part of $Z$, as otherwise, $Z$ would consist of only $k-1$ vertices connected by $k$ edges, contradicting the fact that the out-degree of each fragment is $1$.
On the other hand, if $e$ is not in $Z$, then there must exist a path $f_1,\dots,f_\ell$ where $f_1$ is part of $Z$, none of the $f_i$ ($1<i\le \ell$) are part of $Z$, and $f_\ell$ is an endpoint of the core edge; without loss of generality, assume that $f_\ell=f$.
Since the out-degree of each fragment is $1$ and the outgoing edge incident to $f$ points to $f'$, it follows that such a path cannot exist.
\end{proof}

Lemma~\ref{lem:frags} suggests a simple way for checking whether $C$ contains a bad cycle. 
We describe the following operations on $\mathcal{F}$. 
It is straightforward to translate these operations to the actual network $G$ via broadcasting and convergecasting along the fragment edges.
The fragment leaders first confirm with their neighboring fragments (in $\mathcal{F}$) whether they have an incident core edge, by waiting for $2\alpha\tilde n$ rounds.
Every fragment leader that does not have an incident core edge simply waits by setting a timer of $t =c_2 \alpha\tilde n$ rounds, for a constant $c_2$, chosen sufficiently large such that a broadcast message can reach every fragment leader in $C$.
If there exists a core edge between $f$ and $f'$, then the leader with the greater ID, say $f$, broadcasts a \texttt{merge!} message, which is forwarded to all fragments in $C$ by ignoring the direction of the inter-fragment edges, and is guaranteed to arrive within $t$ rounds at every fragment leader of $C$. 
Upon receiving this message, every node in $C$ adopts $f$'s ID as its new fragment ID.
Note that all fragments will start the next iteration in the exact same round.

On the other hand, if the fragments form a bad cycle, then, after $t$ rounds, all leaders in $C$ conclude that there is no core edge in $C$.
Lemma~\ref{lem:frags} ensures that there must be a bad cycle. 
Thus, the nodes in $C$ do not receive a \texttt{merge!} message, causing them to output $0$, and terminate at the end of this iteration.

\enlargethispage{\baselineskip}
\paragraph*{Check Size Requirement:}
Eventually, in some iteration, it may happen that the fragment leader  $u_\ell$ receives nil messages from all its children, which means that none of the nodes in the fragment has any unexplored edges left.
(Note that this also includes the special case where a node does not have any incident edges of $T$.)
In that case, $u_\ell$ initiates counting the number of nodes in the fragment via a simple broadcast and convergecast mechanism. 
Once the counting process is complete, the root $u_\ell$ outputs $0$ if the fragment contains less than $\frac{\tilde{n}}{\alpha}$ nodes, %
and it disseminates its output to all fragment nodes who in turn output $0$ and terminate. 
Otherwise, $u_\ell$ instructs all nodes to output 1. 

\onlyLong{
\subsection{Proof of Theorem~\ref{thm:mst_upper}} 

The correctness of merging fragments as well as broadcasting and converge casting in the individual fragments is straightforward since the network is synchronous, all nodes are awake initially, and there is a global clock. 
We refer the reader to \cite{peleg2000distributed} for a detailed analysis of the algorithm of  \cite{gallager1983distributed} in the synchronous model.
Here, we restrict our attention to the relevant modifications and the claimed bound on the message complexity.

We first show that the algorithm satisfies the $\OneDetects$ requirement when nodes are given an $\alpha$-approximation of the network size.
Consider the case that $T$ is indeed a spanning tree of $G$.
We claim that the verification algorithm will arrive at the same tree  and that all nodes output $1$.
As we only consider edges from $T$ when adding outgoing edges to a fragment and $T$ is acyclic, it is clear that nodes never detect any cycles when growing fragments.
This means that any component $C$ of the fragment graph $\mathcal{F}$ (defined in Section~\ref{subsec: ST_algo}) consists of a directed tree of fragments and, by Lemma~\ref{lem:frags}, we know that there is a single core edge between two fragments $f$ and $f'$ in $C$.
According to the algorithm, either $f$ or $f'$ will broadcast a \texttt{merge!} message that is forwarded to all fragment leaders in $C$.
Since all fragment leaders wait for a number of rounds that is sufficient for them to receive a potentially arriving \texttt{merge!} message, they receive this message before starting the next iteration, which ensures that all fragments in $C$ merge and continue.

Furthermore, as $T$ contains all nodes of $G$, it follows that we eventually obtain a single fragment. 
Hence, the \textit{Check Size Requirement} of the algorithm will be satisfied, which proves the claim. 

Now suppose that $T$ is \emph{not} a spanning tree of $G$. 
We distinguish two cases:
\begin{enumerate}
\item $T$ contains a cycle $Z$:
First, recall from the algorithm description that a cycle in the same fragment results in an \texttt{illegal} message that is forwarded to the fragment leader, prompting all nodes in the fragment to terminate.

Next, consider the case where the cycle $Z$ spans multiple fragments.
Since the fragments grow and merge by exploring edges of $T$, eventually, there must exist a fragment $F$ that contains all but one edge $e$ of $Z$, and the next edge that is chosen to be explored by some node in $F$ is $e$. 
Let $C$ be the component of the fragment graph $\mathcal{F}$ (defined in Section~\ref{subsec: ST_algo}) that contains $F$. 
According to Lemma~\ref{lem:frags}, the subgraph of $\mathcal{F}$ formed by the fragments of $C$ cannot contain any core edge, and hence every fragment leader will conclude (after its timeout of $t$ rounds expires) that there is indeed a cycle.
This means that the algorithm fails the \textit{Acyclicity Check}, and all nodes that are in some fragment in $C$ output $0$.

\item $T$ does not span $G$:
This means that the labeled edges induce at least two disjoint subgraphs $S$ and $S'$; without loss of generality, assume that $|S| \le |S'|$.
Since the nodes in $S$ faithfully execute the algorithm, at least one of the nodes will output $0$ if $S$ contains any cycle (as argued above). 
Given that 
\begin{align*}
  |S| = n - |S'| \le n - |S|,
\end{align*}
it follows that
$
|S| \le \frac{n}{2}. 
$
Moreover, by assumption $\alpha < \sqrt{2}$, and thus
\begin{align}
 \frac{\tilde{n}}{\alpha} \ge\frac{{n}}{\alpha^2} > \frac{n}{2} \ge |S|.  \label{eq:mst_upper1}
\end{align}
Thus, at the point when all nodes in $S$ are in the same fragment $F$ and there are no more unexplored edges of $T$ incident to nodes in $F$, the fragment leader will check if $|S| \ge \frac{\tilde{n}}{\alpha}$. Inequality \eqref{eq:mst_upper1} provides the necessary contradiction.
\end{enumerate}

(The intuition behind \eqref{eq:mst_upper1} is that, when allowing $\alpha \ge \sqrt{2}$, it would be possible to arrive at the case where all fragments have no unexplored edges and each of the fragments would have a  size of at least $\frac{\tilde{n}}{\alpha}$; hence, none of them would be able to detect that $T$ is disconnected.) %

Finally, to see that the algorithm satisfies $\AllDetect$ when $\alpha = 1$, note that if $T$ forms multiple components, then the \textit{Size Requirement} check will fail for all fragments. 
Thus, $T$ must span $G$ and it follows along the lines of the above analysis that we eventually obtain a single fragment $F$ containing all nodes in $T$. 
Clearly, if $T$ contains a cycle, at least one node in $F$ will notice and inform all other nodes in $G$ via the fragment leader.

Next, we analyze the message complexity of the algorithm. 
Since each fragment forms a tree, all broadcast and convergecast communication between the fragment leaders and the other members of the fragments incurs at most $O(n)$ messages. 
In contrast to \cite{gallager1983distributed} and the state of the art algorithm under the $\ktzero$ assumption of \cite{DBLP:journals/jacm/Elkin20}, our algorithm avoids the communication-expensive operation of exchanging the fragment IDs between nodes.
Instead, when growing a fragment, we select an arbitrary edge $e$ from $T$ that is incident to some vertex of the fragment, and if $e$ closes a cycle, the fragment leader is informed and all nodes in the fragment terminate, i.e., stop sending messages.
Hence the \emph{communication graph}, which is the graph induced by the messages of the algorithm on the vertices of $G$ has at most one cycle per fragment.
Recall that each iteration takes $O(\alpha\tilde{n})=O(n)$ rounds.  
This ensures that every fragment that can merge, merges with at least one other fragment in each iteration, and hence a standard analysis proves that the number of fragments reduces by half, i.e., there are $O(\log n)$ iterations in total.
Each iteration may involve at most $O(n)$ messages due to broadcast/convergecast communication between the fragment leaders and the nodes in their fragments, and the broadcast of possible \texttt{merge!} messages.
It follows that the overall message complexity is $O(n\log n)$.
This completes the proof of Theorem~\ref{thm:mst_upper}.
}
\section{Algorithms for Verifying a $d$-Approximate BFS Tree} \label{sec:bfs_algo}

We now turn our attention to the $d$-approximate BFS tree verification problem. 
The following lemma suggests that we can extend the algorithm for ST verification described in Section~\ref{sec:st_algo} by inspecting the neighborhood of the root when considering a sufficiently large stretch $d$:

\begin{lemma}\label{lem:large_stretch}
Let $T$ be a spanning tree of $G$ with root $r$. Let $d \geq \frac{n-1}{x+1}$, for some integer $x \ge 1$. 
If $\dist_T(r,u) > d \cdot \dist_G(r,u)$ for some node $u$ in $G$, then $\dist_{G}(r,u) \le x$.
\end{lemma}

\begin{proof}
Consider a node $u$ such that  $\dist_T(r,u) > d \cdot \dist_G(r,u)$. It holds that
\begin{align}
    \dist_{G}(r,u) < \frac{\dist_{T}(r,u)}{d} %
                   \le \frac{n-1}{d} %
                   \le x+1,\notag 
\end{align}
and thus $\dist_{G}(r,u) \le x$. 
\end{proof}

\begin{theorem} \label{thm:bfs_algo}
Consider the $\kt_\rho$ assumption, for any $\rho \ge 0$.
If $d \geq \frac{n-1}{\max\{2,\rho+1\}}$, there exists a deterministic algorithm for $d$-approximate BFS verification that satisfies $\OneDetects$ with a message complexity of $O(n \log n)$ and a time complexity of $O(n\log n)$ rounds, assuming that nodes are given an $\alpha$-approximation of the network size, for some $\alpha<\sqrt{2}$.
If nodes have perfect knowledge of the network size, the algorithm ensures $\AllDetect$.
\end{theorem}
We point out that there is no hope of getting $O(n\log n)$ messages for  significantly smaller values of $d$, as the lower bound in Theorem~\ref{thm:ktx_bfs_lb}  holds for any 
$d \le \frac{(1-\beta)n}{4\rho-2}$ where  $0<\beta<1$ is a constant.

\begin{proof}%
First consider the case $\rho \in \set{0,1}$.
Instantiating Lemma \ref{lem:large_stretch} with $x=1$, tells us that
we only need to check if $T$ is a spanning tree and that all edges incident to the root are in $T$ in order to verify if a subgraph $T$ is a $d$-approximate BFS. 
More concretely, after executing the spanning tree verification, each node computes its distance in $T$ from the root. 
Then the root directly contacts all its neighbors (in $G$):
If there is a node at distance at least $\frac{n-1}{2}$ from the root (in $T$) who is not a neighbor of the root, it outputs $0$, and broadcasts a \texttt{fail} message to all nodes in $T$; otherwise, it broadcasts an \texttt{accept} message.
Each node in $T$ decides accordingly once it receives this message from its parent.

The argument for the case $\rho \ge 2$ is similar, except that we now instantiate Lemma~\ref{lem:large_stretch} with $x=\rho$.
That is, since the root knows its $\rho$-hop neighborhood, it can contact all nodes within distance $\rho$ (in $G$) by using only $O(n)$ messages. 
In more detail, the root locally computes a spanning tree $T_\rho$ of its $\rho$-neighborhood, and perform a standard flooding algorithm on $T_\rho$.  
\end{proof}

\section{Discussion and Open Problems}
In this paper, we study the message complexity of ST verification, MST verification and $d$-approximate BFS verification distributed algorithms.
To the best of our knowledge, the message complexity of $d$-approximate BFS verification distributed algorithms has never been studied before, hence, we focus our discussion on this problem. 
In our study, we show that the message complexity is largely determined by the stretch $d$. 
When $d$ is small, we obtain a message complexity lower bound of $\Omega(n^2)$ for $\ktzero$ and $\ktone$, and a lower bound of $\Omega\left(\frac{1}{\rho}\left(\frac{n}{\rho}\right)^{1+\frac{c}{\rho}}\right)$ (for some constant $c >0$) for $\ktx$ where $\rho > 1$. 
The bound of $d$ is almost tight for $\ktzero$ model, but not for $\ktx$. 
In particular, for $\ktx$ model where $\rho > 1$, it is still open whether we can match the lower bound that holds for $d < \frac{n-1}{\max\set{2, \rho+1}}$. 

In addition, all the bounds we obtain for $\ktx$ where $\rho \ge1$ are restricted to comparison-based algorithms, while the bound for $\ktzero$ holds for general algorithms. 
This gives rise to the following important unanswered question: Can the lower bound of $\Omega(n^2)$ for $\ktone$ be improved by using non-comparison based algorithms?

\input{main_fixed.bbl}

\end{document}

%% file: main_fixed.bbl
%% BioMed_Central_Bib_Style_v1.01